\documentclass[11pt]{article} \date{}
\usepackage{latexsym,amsfonts,amssymb,amsthm,amsmath,xcolor,authblk}

\newtheorem{theorem}{Theorem}
\newtheorem{lemma}[theorem]{Lemma}
\newtheorem{proposition}[theorem]{Proposition}
\newtheorem{corollary}[theorem]{Corollary}

\theoremstyle{remark}

\newcommand{\Ai}{\mathop{\rm Ai}}

\begin{document}

\title{Resonances in the one dimensional Stark effect in the limit of small field}

\author{Richard Froese and Ira Herbst}

\maketitle

\begin{abstract}
We discuss the resonances of Hamiltonians with a constant electric field in one dimension in the limit of small field.  These resonances occur near the real axis, near zeros of the analytic continuation of a reflection coefficient for potential scattering, and near the line $\arg z = - 2\pi/3$.  We calculate their asymptotics.  In conclusion we make some remarks about the higher dimensional problem.
\end{abstract}

\tableofcontents

\section{Introduction}
The purpose of this paper is to calculate the asymptotics of the resonances of the operator $H_f = -d^2/dx^2 + V(x) + fx$ for small $f> 0$.  Here $V$ is a real bounded potential of compact support.  The potential $fx$ represents an electric field of strength $f$ in the negative $x$ direction acting on a particle of charge $1$ and mass $1/2$. There have been several simple models considered which contribute to the understanding of the existence and non-existence of resonances with an emphasis on the fate of pre-existing resonances of $H =  -d^2/dx^2 + V(x) $  (see \cite {HR, HM, JY}) as $f \to 0$.  None of these references treat the model given by $H_f$ above although in \cite{HM} a specific model is treated where $V $ is replaced by the sum of two delta functions.  In this paper we are not only interested in pre-existing resonances which mostly disappear in the limit $f \to 0$, rather we attempt to calculate the limiting behavior of all the resonances of  $H_f$ for small $f > 0$ (at least those in a compact set of the plane).  We would like to mention papers \cite{EK1, EK2} which treat the resonances of $H_f$ for $f=1$ at high energy with a compact support $V$ satisfying some assumptions.  There is a relation between high energy with $f=1$ and small $f$ with fixed energy but with a much changed potential $V$.

\vspace {.1in}

We define a resonance as a pole in the meromorphic continuation of the resolvent $(H_f - z)^{-1}$ in the open lower half plane.  Sometimes we will refer to $k= \sqrt z$ as a resonance when $z$ is a resonance.  Our results show that as $f \to 0 $ resonances are of four different types in any compact set of the closed lower half plane:\\

\begin{flushleft}
1) Resonances crowd densely along the positive real axis converging to this axis. We compute these in Section \ref{pra}.\\
\vspace{.1in}

2) Resonances crowd densely along the line $\arg z = - 2\pi/3$.  These are computed in Section \ref{line}.\\
\vspace{.1in}

3) Resonances converge to points $z$ where an analytic continuation of a certain reflection coefficient for scattering with the potential $V$ vanishes.  This occurs only when in addition $-2\pi/3 < \arg z < 0$ and these points are the only limits of resonances of $H_f$ in this sector.  This is discussed in Section \ref{rf}. \\
\vspace{.1in}
4) Resonances move continuously into the lower half plane from the negative eigenvalues of $H$ as $f$ increases from $0$. The imaginary parts are exponentially small as $f \to 0$.  They are computed in Section \ref{rnnra}.  For small $f$ there are no resonances in the region $ - \pi + \delta < \arg z < -2\pi/3 - \delta$ for $\delta > 0$ (see Theorem \ref{noresonances} for a finer description of this region).   
\end{flushleft}
We also have some results for higher dimensions but they represent only a start on the problem.  See Section \ref{ndimensions}.

\section{Resonance - free regions}

We write $V = V_1V_2$ with $V_j $ real and bounded and $ H_{f,0} = H_f - V$.  An easy computation verifies that for $\text{Im} z \ne 0$,

$$ I - V_1 (H_f -z)^{-1}V_2 = (I + V_1 (H_{f,0} -z)^{-1}V_2)^{-1}.$$
It is easy to see that $K_{0,f}(z): = V_1(H_{f,0}-z)^{-1}V_2$ has an entire analytic continuation from the upper half plane to $\mathbb{C}$ (see (\ref{K}) and the discussion below) and thus since the analytic continuation of $K_{0,f}(z)$ is compact $I - V_1 (H_f -z)^{-1}V_2$ has a meromorphic continuation to  $\mathbb{C}$. We define a resonance as a pole of this  meromorphic continuation. \\ 
Note the resolvent equation 
\begin{equation}\label{resolventeqn}
(H_f-z)^{-1} = (H_{f,0} -z)^{-1} - (H_{f,0} -z)^{-1}V_2(1+  K_{0,f}(z))^{-1}V_1 (H_{f,0} -z)^{-1}
\end{equation}
from which we can see, in particular, that if there is a pole in the meromorphic continuation of $ V_1 (H_f -z)^{-1}V_2$ then there is a pole in  the meromorphic continuation of the integral kernel of $(H_f -z)^{-1}$ and vice versa.\\
Using the Airy function defined as $\Ai(x) = \lim_{R \to \infty} (2\pi)^{-1}\int _{-R}^R e^{i(k^3/3 + kx)}dk$,  it follows that the unitary operator, $e^{\pm i p^3/3}$ ($p=-id/dx$), has an integral kernel:

\begin{equation}
\exp(\pm ip^3/3)f(x) = \int \Ai(\pm(x-y)) f(y)dy.
\end{equation}
We obtain the kernel $K_{0,f}(x,y;z)$ of the operator $K_{0,f}(z)$  for $\text{Im} z \neq 0$ using the identity $e^{ip^3/3f} fx e^{-ip^3/3f} = p^2 + fx$,
\begin{align} \label{K}
K_{0,f}(x,y;z)= f^{-1/3}V_1(x) \Big(\int \Ai(f^{1/3}(x-t))(t-z/f)^{-1}\Ai(f^{1/3}(y-t))dt\Big) V_2(y).
\end{align}
Using the fact that the Airy function is entire we can analytically continue the operator $K_{0,f}(z) = V_1(H_{f,0} - z)^{-1}V_2$ into the lower half plane.  This can be accomplished by distorting the contour in (\ref{K}) locally into the lower half plane in a neighborhood of $ \text{Re} (z/f)$ and picking up a pole term so that the analytic continuation  of $V_1(H_{f,0} - z)^{-1}V_2$ into the lower half plane has kernel $K_{0,f,c}(x,y;z)$ given by  

\begin{align} \label{K_c}
K_{0,f,c}(x,y;z) = 2\pi i f^{-1/3}g_1(x;z)g_2(y;z) + K_{0,f}(x,y;z)
\end{align}
where in (\ref{K_c}) $\text {Im z} < 0$ and 

\begin{equation}
g_j(x;z) = V_j(x) \Ai(f^{1/3}(x-z/f))
\end{equation}
With the notation  $(V_1(H_{f,0} - z)^{-1}V_2)_c$ for the analytic continuation of  $V_1(H_{f,0} - z)^{-1}V_2$ into the lower half plane we have,

\begin{align*} 
 & I + (V_1(H_{f,0} - z)^{-1}V_2)_c = I+  Q  + V_1(H_{f,0} - z)^{-1}V_2  \label {OP_c}\\ 
& Q =  (e_2, \cdot )e_1 \\
&  e_1 = 2\pi i f^{-1/3} g_1\\
&  e_2 = \overline g_2.
\end{align*}
The operator with integral kernel $K_{0,f,c}$ has eigenvalue $-1$ when $z=z_0$ in the lower half plane if and only if $z_0$ is a resonance.  This follows from the Fredholm alternative.
Thus the resonances are points $z$ in the lower half plane for which there is a non-trivial solution to 
\begin{align}
(I + Q + K_{0,f}(z))\psi = 0.
\end{align}
This equation has a solution only if $(I + K_{0,f}(z))\psi =-Q\psi = ce_1$.  Substituting back into
\begin{align*}
&\Big((I + Q(I + K_{0,f}(z))^{-1}\Big)(I+K_{0,f}(z))\psi = 0\\
\end{align*} 
shows that we can have $c\ne 0$ if and only if $1+ (e_2,(I+K_{0,f}(z))^{-1}e_1) = 0$  which gives 

\begin{theorem}\label{G}
Define $A_2(x) = \Ai(f^{1/3}(x-z/f)), e_1(x) = 2\pi i f^{-1/3}V_1(x)A_2(x), e_2(x) = V_2(x)\overline{A_2(x)}$ and
\begin{equation}\label{Gfunction}
G(z,f) = 1 +  (e_2,(I-V_1(H_f -z)^{-1}V_2)e_1).
\end{equation}
The equation $G(z,f) = 0$ determines the resonances of $H_f$ in the lower half plane .\\  
\end{theorem}

It is easy to find a formula for the continued resolvent of $H_f$ using a simple formula for $(I -S)^{-1}$ where $S$ is a rank one operator satisfying $S^2 = \alpha S$ for some $\alpha \ne 1$. We have $$ (I-S)^{-1} = I + (1-\alpha)^{-1} S.$$

With the notations $K_f = V_1(H_f -z)^{-1}V_2$ and $K_{f,c} = (V_1(H_f -z)^{-1}V_2)_c$ we obtain $$ K_{f,c} =K_f + \frac{(I - K_f)|e_1\rangle \langle e_2|(I - K_f)}{G(z,f)},$$
where a subscript $c$ indicates meromorphic continuation across the real axis from the upper half plane to the lower half plane.  

Let $k = \sqrt z$  and $\arg k \in [-\pi/2, 0]$.   We need to obtain the asymptotics of $G(z,f)$ for small $f$.  Using expansions in Abramowitz and Stegun \cite{AS} ($10.4.59, 10.4.60$), for $x$ in a compact set and $|k| > \delta_1 > 0$ we have

\begin{align}
\Ai(f^{1/3}(x-z/f)) &= \frac{f^{1/6}}{2\sqrt {\pi k}}e^{(i\zeta_x - i\pi/4)}(1 + O(f)) ; \arg k \in [-\pi/2, -\delta) \label {A_2}\\
= &\frac{f^{1/6}}{2\sqrt {\pi k}}e^{i(\zeta_x -\pi/4)}(1 + O(f) + i e^{-2i\zeta_x}); \arg k \in (-\pi/3 + \delta, 0] \label {A'_2}
\end{align} 
where $\zeta_x = \frac{2k^3}{3f} -kx$. 
We can use (\ref{A'_2}) to conclude that if $ \arg k \in [-\pi/6,0]$ and $\text {Im} k < - \frac{f\log(1/f)}{8|k|^2} +Cf/|k|^2$ then the equality in (\ref{A_2}) also holds.  This follows from the simple estimate $|e^{-4ik^3/3f}| \le e^{-8|k|^2|\text{Im}k|/3f}$ for $\arg k \in [-\pi/6,0]$.

\begin{lemma}\label{rdiff}
Suppose $0< \delta < |k| < \delta^{-1}$ and $\arg k  \in (-\pi/6, 0)$.  Then there is a constant $C>0$ so that if both $f>0$ and $|e^{-4ik^3/3f}|$ are small enough 
$$||V_1(H_f-z)^{-1}V_2 - V_1(H-z)^{-1}V_2|| \le C|e^{-4ik^3/3f}| + O(f).$$
Here $z=k^2$.
\end{lemma}

\begin{proof}
The integral kernel of the resolvent $(H_{0,f}-z)^{-1}$ is for $x<y$ given by 
$$\tilde {A_1}(x)A_2(y)/W$$ where $\tilde A_1(x) = \Ai(e^{-2\pi i/3}f^{1/3}(x-z/f))\in L^2(-\infty,0)$, $A_2(x) = \Ai(f^{1/3}(x-z/f)) \in L^2(0,\infty)$ and the Wronskian $W = f^{1/3}e^{i\pi/6}/2\pi$. We calculate for $x<y$, and for $x$ and $y$ in a compact set 
$$\tilde{A_1}(x)A_2(y)/W = (e^{-ik|x-y|}/2k)[1+O(f) + ie^{-i(4k^3/3f -2yk)}](1+O(f)).$$
It follows that 
\begin{equation}\label{perturb}
    ||V_1(H_{0,f}-z)^{-1}V_2 - V_1(p^2 - z)^{-1}V_2|| \le C_1|e^{-4ik^3/3f}| +O(f).
\end{equation}
We have 
$$(1-V_1(H_f-z)^{-1}V_2) -(1- V_1(H-z)^{-1}V_2) = (1+V_1(H_{0,f}-z)^{-1}V_2)^{-1} -(1+ V_1(p^2-z)^{-1}V_2)^{-1}.$$
We use the perturbation formula 
$$(1+A)^{-1} -(1+B)^{-1} = - (1+B)^{-1}(A-B)(1 + (1+B)^{-1}(A-B))^{-1} (1+B)^{-1}.$$
The result now follows from Eq.(\ref{perturb}) and a bound on $(1+B)^{-1} = 1-V_1(H-z)^{-1}V_2$ near the real axis which follows from limiting absorption estimates.

\end{proof}

Let
$$F(k) =   (V_2,e^{-ikx}(I - V_1 (H-z)^{-1}V_2)e^{-ikx}V_1), \hspace{.1cm} z = k^2.$$
In the next proposition we get a handle on the behavior of $G(z,f)$ not too far from the positive real axis when $|\text{Im}k|/f$ is large enough.  The next proposition will be effective only when $F(k)$ does not vanish. 
\begin{proposition}\label{Gest}
Suppose $0< \delta < |k| < \delta^{-1}$ and $\arg k  \in (-\pi/6, 0)$.  Then
$$G(z,f) = 1 + (e_2,(1-V_1(H_f-z)^{-1}V_2)e_1) = 1 + e^{4ik^3/3f} (2k)^{-1}F(k)+E$$
$$|E|\le C + O(f)|e^{4ik^3/3f}|$$
for some constant $C>0$.
\end{proposition}

\begin{proof}
The proposition follows by inserting the estimates in (\ref{A'_2}) and in Lemma \ref{rdiff} into the definition of $G(z,f)$ given in Theorem \ref{G}.
\end{proof}
Let $d(z) = \text{dist}(z,\sigma(H))$

\begin{proposition}\label{Gest'}
\begin{align}
G(z,f) &= 1 +  (2k)^{-1}e^{\frac{4ik^{3}}{3f}}\left [ F(k)+ O(f)\right ] \label{nopole}\\
&\text{if Im}z < -\delta, |k| < \delta^{-1}. \label{easyG} \\
G(z,f) &= 1 +  (2k)^{-1}e^{\frac{4ik^{3}}{3f}}\left [F(k)+ O(f)\right ] \nonumber\\
& \text {if} \arg k \in [-\pi/6, 0), \text{Im} k \le  -\frac{3f \log1/f}{8|k|^2},\text{and} \ \delta < |k| < \delta^{-1}.  \label{nearposreal} \\
G(z,f) &= 1 +  (2k)^{-1}e^{\frac{4ik^{3}}{3f}}\left [ F(k) + O(f/d(z)^2)\right ]  \label{GnearR_} \\
& \text { if } \arg k \in [-\pi/2, -5\pi/12), \delta < |k| < \delta^{-1}, \text{and} \label{nearnegreal} \nonumber \\
&  d(z) > C_1f \ \text{where} \ C_1 \ \text{is large enough}. 
\end{align} 
\end{proposition}

\begin{proof}
The first estimate follows directly from (\ref{A_2}) and the fact that if $\text {Im} z < -\delta$, $||V_1((H_f - z)^{-1} - (H-z)^{-1} )V_2|| \le Cf$ follows easily.  
The second estimate again follows from the equation in (\ref{A_2}) which is valid when $\arg k \in [-\pi/6,0]$ as long as $ \text{Im} k \le -\frac{f \log1/f}{8|k|^2}$ and $|k| > \delta$ and from Lemma \ref{rdiff}.  The latter gives  $$ ||(I - V_1(H_f - z)^{-1}V_2) - (I - V_1(H-z)^{-1}V_2)|| \le Cf.$$ for these $k$.
The third estimate follows from the fact that $$||V_1(H_{0,f} - z)^{-1}V_2 - V_1(p^2 -z)^{-1}V_2|| \le Cf$$ for $\text {Re}z < - \delta$, $\arg k \in [-\pi/2, -5\pi/12]$ (one can use the explicit asymptotic behavior of Airy functions or see Section \ref{resconv}).  Writing $K_f = V_1(H_{0,f}-z)^{-1}V_2$ and $K_1 = V_1(p^2 -z)^{-1}V_2$, We have $$I + K_f = (I + K_1)(I- (I+K_1)^{-1} (K_1 - K_f)).$$  So if we let $B_f = (I+K_1)^{-1} (K_1 - K_f)= (I - V_1(H-z)^{-1}V_2)(K_0 - K_f)$, we have $||B_f|| \le Cfd(z)^{-1}$ (where we have assumed for simplicity that $|k| < \delta^{-1}$).  It follows that $V_1(H-z)^{-1}V_2 - V_1(H_f - z)^{-1}V_2 = B_f(I-B_f)^{-1} (I - V_1(H-z)^{-1}V_2)$.  It follows that if $fd(z)^{-1}$ is small enough (so that $||B_f|| <1/2$ for example), we have  $||V_1(H-z)^{-1}V_2 - V_1(H_f - z)^{-1}V_2|| \le Cf/d(z)^2$. 
\end{proof}
The next theorem is a corollary of  Proposition \ref{Gest'}, Proposition \ref{Gest}, and the behavior of $|e^{4ik^3/3f}|$ as $f\downarrow 0$.  Note the definition $d(z) = \text{dist}(z,\sigma(H))$. 

\begin{theorem} \label{noresonances}
\begin{enumerate}
\item[(i)]
Given $\delta > 0$ and $C_1>0$, if $C_0$ is large enough then for small $f>0$, $H_f$ has no resonances in the set $\{z: \delta^{-1} > |z|>\delta >0, \arg k \in [-\pi/6 , 0), \text{Im}k \le - C_0f, |F(k)| > C_1 \}$.
\item [(ii)] 
Given $\delta>0$, if $C_0$ is large enough then for small $f>0$, $H_f$ has no resonances in the set $\{z: \delta^{-1} > |z|>\delta >0, \arg k \in [-\pi/6 , 0), \text{Im}k \le -\frac{3f \log1/f}{8|k|^2}, |F(k)| > C_0f \}$.

\item [(iii)] 
Given $\delta > 0$,  if $C_0$ is large enough then for small $f > 0$, $H_f$ has no resonances in the set $\{z: \delta^{-1} > |z|>\delta >0, \arg z  \in [-\pi, -5\pi/6], d(z) >C_0f \}$.
\item [(iv)] 
Given $\delta > 0$, if $C_0$ is large enough then for small $f>0$, $H_f$ has no resonances in the set $\{k: |k| <\delta^{-1}, \arg k  \in [-\pi/2 + \delta, -\pi/3], k = e^{-i\pi/3}(k_0 - i\kappa), k_0 > \delta, \kappa > C_0f\}$. 
\item [(v)] 
Given $\delta > 0$, if $C_0$ is large enough then for small $f>0$, $H_f$ has no resonances in the set $\{k: |k| <\delta^{-1}, \arg k  \in [-\pi/3, -\delta], k = e^{-i\pi/3}(k_0 + i\kappa), k_0 > \delta >0, \kappa > C_0f, |F(k)| > \delta \}$. 
\item [(vi)]
Given $\delta > 0$, if  $C_0$ is large enough then for small $f>0$, $H_f$ has no resonances in the set $\{k: |k| <\delta^{-1}, \arg k  \in [-\pi/3, -\delta], k = e^{-i\pi/3}(k_0 + i\kappa), k_0 > \delta >0, \kappa > \frac{f \log(1/f)}{8|k|^2}, |F(k)| > C_0f \}$. 
\end{enumerate}
\end{theorem}

\begin{proof}
(i) and (ii) follow from Proposition \ref{Gest} using $|e^{4ik^3/3f}| \ge e^{8|k|^2|\text{Im}k|/3f}.$ 

In the proof of (iii), we use (\ref{GnearR_}) and the fact that for $k$ in the stated set $|e^{\frac{4ik^{3}}{3f}}|  <  e^{-2\sqrt{2}|k|^3/3f}$ so if $C_0 $ is large enough we have $|G(z,f)| > 1/2$.

In the proof of (iv) we have  $|e^{\frac{4ik^{3}}{3f}}| \le  e^{-8|k|^2\kappa/3 f}$ so that if $C_0$ is large enough $|G(z,f)| > 1/2$ for small $f>0$.

The proofs of  (v) and (vi) are very similar to those of (i) and (ii).

\end{proof}

\section{Vanishing of the reflection coefficient}\label{rf}

In theorem \ref{noresonances} in the region $\{k: \arg k \in (-\pi/3,0) \}$ we demanded the non-vanishing of $F(k)$ to conclude the absence of resonances.
According to \cite{T}, p.139, the S-matrix for potential scattering from momentum $k \in \mathbb{R}$ to momentum $k' \in \mathbb{R}$ is given as 
\begin{equation}\label{STdef}
\langle{k'}|S|k\rangle = \delta(k'-k) -2\pi i\delta(k'^2 - k^2) \langle{k'}|T(k^2 + i 0)|k\rangle
\end{equation}
where 
$$T(z) = V-V(H-z)^{-1}V$$
is the ``off-shell" T-matrix.  Noting that $|k\rangle = e^{ikx}/\sqrt{2\pi}$ we see that $F(k) = (V_2,e^{-ikx}(I-V_1(H-k^2)^{-1}V_2e^{-ikx})$ is the analytic continuation 
of $2\pi \langle{-k}|T(k^2 + i 0)|k\rangle$ from $k = \sqrt z$ around the branch point $z=0$ to the region we are interested in, $-2\pi/3 < \arg z < 0$. Note that as $z$ makes a $2\pi$ revolution, $k$ changes sign. Thus the vanishing of $F(k)$ is the vanishing of the analytic continuation of the amplitude for reflection for an incoming particle with momentum  $k$.

We have learned that if $k_0 = \sqrt z_0 $ with $\arg k \in (-\pi/3, 0)$ and $F(k_0) \ne 0$, then $z_0$ is not a limit of resonances of $H_f$ as $f \to 0$.  On the other hand, if $F(k_0) = 0$ and  $\arg k \in (-\pi/3, 0)$ then generically $k_0$ will be the limit of resonances of $H_f$.  We state this as a theorem:

\begin{theorem} If $k_0 \ne 0$ with $-\pi/3< \arg k_0 < 0$,  and $$F(k) : = (V_2, e^{-ikx}\left(I- V_1(H-k^2)^{-1}V_2\right)e^{-ikx}V_1)$$ is $0$ at $k_0$ while the derivative $F'(k_0) \ne 0$ then as $f \downarrow 0 $, $k_0$ is a limit of resonances of $H_f$.  If  $k_0 \ne 0$ with $-\pi/3 < \arg k_0 < 0$ and $F(k_0) \ne 0$ or if $k_0 \ne 0$ with $-\pi/2 < \arg k_0 < -\pi/3$ then there are no resonances of $H_f$ near $k_0$ for $f>0$ small.  
\end{theorem}

Note that there is more detailed information about the absence of resonances in Theorem \ref{noresonances}.

\begin{proof}
After the discussion above what is left to prove is that if $k_0 \ne 0, -\pi/3 < \arg k_0 < 0, F(k_0) =0, F'(k_0) \ne 0$, then $k_0$ is a limit of resonances of $H_f$.

Let 

\begin{align}  
&H(k,f) = 2ke^{-4ik^3/3f} + 2ke^{-4ik^3/3f}(e_2,(I - V_1(H_f-z)^{-1}V_2)e_1) \nonumber \\
& = 2ke^{-4ik^3/3f} + (V_2, a(x,k,f)(I- V_1(H_f-z)^{-1}V_2)a(x,k,f)V_1).
\end{align}
where 

\begin{align}
&a(x,k,f) = \sqrt{4\pi k}e^{i\pi/4}f^{-1/6}e^{-2ik^3/3f}\Ai(f^{1/3}(x-z/f)) \to e^{-ikx}
\end{align}
as $f \downarrow 0$.  We want to solve $H(k,f) = 0$ for $k$ near $k_0$ given $f> 0$ and small. We can write 

\begin{align}\label{AvronHerbstformula}
V_1(p^2 + fx -k^2)^{-1}V_2 =-i \int_0^\infty V_1e^{itfx/2}e^{i(p^2 - k^2)t}e^{itfx/2}V_2e^{it^3f^2/12}dt
\end{align}
if $-\pi/2< \arg k < 0$.  We thus see that this quantity is $\mathcal {C}^{\infty}$ in the variable $f$ for $f \in \mathbb{R}$.  Since  $I - V_1 (H_f -z)^{-1}V_2 = (I + V_1 (H_{f,0} -z)^{-1}V_2)^{-1}$, it follows that $I - V_1 (H_f -z)^{-1}V_2$ is   $\mathcal {C}^{\infty}$ for small $f$.  The function $a(x,k,f)$ is at least $\mathcal {C}^\infty$ in $(k,f)$
in a region of the form $B_r(k_0) \times [0,f_0)$ where $B_r(k_0)$ is a small ball centered at $k_0$ where $H(k_0,0)=0$ (see the Appendix). Here and in the following we define $e^{-4ik^3/3f}$ and all its derivatives to be zero at $f=0$.  This makes the latter function $\mathcal {C}^{\infty}$ for $f\in [0,\infty)$ and $k$ near $k_0$.  To make reference to the standard implicit function theorem directly we can consider $H(k,f^2)$ instead which, as is required, is $\mathcal {C}^1$ in $(k,f)$ for $(k,f)$ in an open set containing $(k_0,0)$.  Our assumption is that $H(k_0,0) =  F(k_0) = 0$, but its $k$ derivative $F'(k_0)$ is non-zero at this point.  Thus by the implicit function theorem, for small $f$, $H_f$ has a resonance near $k_0$.  
\end{proof}

For reference we have 

\begin{align} 
F'(k_0) =& -2i \int V(x) e^{-2ik_0x}xdx\nonumber\\
&+ i (V, e^{-ik_0x}R_0e^{-ik_0x}xV) + i (xV, e^{-ik_0x}R_0e^{-ik_0x}V)\nonumber\\
&-2k_0 (V, e^{-ik_0x}R_0^{2} e^{-ik_0x}V)
\end{align}
where $R_0 = (H-k_0^2)^{-1}$.

\section{Resolvent Convergence}\label{resconv}

The purpose of this section is to prove that in a complex neighborhood of a point on the negative real axis we have convergence of the analytically continued resolvent $(p^2  + fx  - z)^{-1}$ to $(p^2-z)^{-1}$ with exponential weights.  

We use the method of E. Mourre \cite{M} as expounded in Perry-Sigal-Simon \cite{PSS} in proving bounds on $(p^2  + fx  - z)^{-1}$ with weights for $\text{Re}z $ in some compact set of the negative reals, uniformly for $\text{Im} z \ne 0$.  

\begin{lemma}
Let $D = (|p|^2 + 1)^{-1/2}$.  Then for $f > 0$ and $\text{Im} z \ne 0$,
$$||D(p^2 + fx - z)^{-1}D|| \le Cf^{-1}$$ where $C$ is a universal constant.
\end{lemma}

\begin{proof}
We take $\text{Im} z \le  0$. Let $G_{\epsilon}(z) = (p^2 + fx + i\epsilon f - z)^{-1}$ and $F_{\epsilon}(z) = DG_{\epsilon}(z)D$.  Clearly   $||G_{\epsilon}(z)|| \le (\epsilon f)^{-1}$.  Note that $f =[ip,p^2 + fx]$.  We have
$$||G_{\epsilon}(z)\phi||^2 = (\phi, G_{\epsilon}(z)^ *2\epsilon f G_{\epsilon}(z)\phi)/2\epsilon f \le  (\phi, G_{\epsilon}(z)^* (2\epsilon f  - 2 \text{Im}z)G_{\epsilon}(z)\phi)/2\epsilon f$$
$$=(-i/2\epsilon f)(\phi, (G_{\epsilon}(z)^*- G_{\epsilon}(z))\phi) \le |(\phi, G_{\epsilon}(z)\phi)|/\epsilon f$$

Thus $$||G_{\epsilon}D|| \le (\epsilon f)^{-1/2} ||DG_{\epsilon}D||^{1/2} = (\epsilon f)^{-1/2} ||F_{\epsilon}||^{1/2}.$$  Similarly $$ ||DG_{\epsilon}||\le (\epsilon f)^{-1/2} ||F_{\epsilon}||^{1/2}.$$ Then
 $$dF_{\epsilon}/d\epsilon  = D (p^2 + fx + i\epsilon f - z)^{-1}[p,p^2 + fx + i\epsilon f - z] (p^2 + fx + i\epsilon f - z)^{-1}D.$$ Thus
 $$||dF_{\epsilon}/d\epsilon|| \le ||DG_{\epsilon}|| + ||G_{\epsilon}D|| \le 2 (\epsilon f)^{-1/2} ||F_{\epsilon}||^{1/2}$$
 Starting from $||F_{\epsilon}|| \le (\epsilon f)^{-1}$ and iterating, we get the result for some universal constant $C$.

\end{proof}

We can now show the convergence of $\langle x \rangle^{-1}(p^2 +fx -z)^{-1}\langle x \rangle^{-1}$ to  $\langle x \rangle^{-1}(p^2 -z)^{-1}\langle x \rangle^{-1}$ uniformly for $\text{Im} z \ne 0$ and $\text{Re} z \le -\delta$ for $\delta >0$.

\begin{lemma}\label{resolventconvergence}
Suppose $\text{Im} z \ne 0$ and $\text{Re} z \le -\delta$ with $\delta >0$.Then 
$$||\langle x \rangle^{-1}(p^2 +fx -z)^{-1}\langle x \rangle^{-1} - \langle x \rangle^{-1}(p^2-z)^{-1}\langle x \rangle^{-1}|| \le Cf $$
where C depends only on $\delta$.

\end{lemma}

\begin{proof}
We use the resolvent equation twice to obtain
\begin{align}
&\langle x \rangle^{-1} \big((p^2 +fx -z)^{-1} - (p^2-z)^{-1}\big)\langle x \rangle^{-1} =  - \langle x \rangle^{-1} (p^2-z)^{-1}fx(p^2-z)^{-1}\langle x \rangle^{-1} \nonumber\\
& +  \langle x \rangle^{-1} (p^2-z)^{-1}fx(p^2 + fx - z)^{-1}fx(p^2-z)^{-1}\langle x \rangle^{-1}.\nonumber
\end{align}
Using the easily established bound
$$||(p + i)x(p^2 -z)^{-1}\langle x \rangle^{-1}|| \le C',$$
with $C'$ dependent only on $\delta$, and the previous lemma we obtain the result.  
\end{proof}

By iterating the equation above we can get an asymptotic expansion for $(p^2 +fx -z)^{-1}$ with appropriate weights.  To shorten our equations we use the abbreviations $W = -x, R_0 = (p^2 -z)^{-1}, R_{0,f} = (p^2 + fx -z)^{-1}$.

\begin{lemma}\label{asymptoticexpansion}
Suppose the functions $\zeta_1$ and $\zeta_2$ are bounded and measurable with $|x^n\zeta_j(x)|\le a_n < \infty$ for all $n$.  Suppose $\text{Re}z < - \delta$ with $\delta > 0$ and $\text{Im}z \ne 0$.  Then for any $N >0$
$$\zeta_1R_{0,f}\zeta_2 = \sum_{n = 0}^{2N-1}(\zeta_1R_0(WR_0)^n\zeta_2) f^n+ (\zeta_1(R_0W)^NfR_{0,f}(WR_0)^N\zeta_2) f^{2N-1}$$
with the remainder $||\zeta_1(R_0W)^NfR_{0,f}(WR_0)^N\zeta_2|| \le b_N $ where $b_N$ depends only on $N, \delta, \zeta_1, and  \ \zeta_2$. 
\end{lemma}

\begin{lemma}\label{jumpconvergence}
Suppose $z=-\alpha - i\beta$ with $\alpha \ge \delta > 0$ and $\alpha \ge 8|\beta|$.  Choose $\epsilon > 0$ so that $\epsilon \sqrt \alpha \ge 8|\beta|$.  Let $g_{\epsilon, z} (x) = e^{-\epsilon |x|} \Ai(f^{1/3}(x-z/f))$.  Then
$$||g_{\epsilon, z}||_2 \le C_{\epsilon}e^{-\epsilon\alpha/4f}$$ for small $f>0$.

\end{lemma}

\begin{proof}
We first write $x+\alpha/f = a, \beta/f = b$.  We obtain 

$$\text{Re}(a+ib)^{3/2} = (a + \sqrt{a^2 + b^2})^{1/2} (2a -  \sqrt{a^2 + b^2})/\sqrt2, |\text{Im}(a+ib)^{3/2}| = ( \sqrt{a^2 + b^2}- a)^{1/2}( \sqrt{a^2 + b^2} +2a)/\sqrt2$$
If the argument $w= f^{1/3}(a+ib)$ of $Ai$ is bounded, then for example we must have 
$$|x| \ge \alpha/f - |a| \ge f^{-1/3}(f^{1/3}\alpha/f -  f^{1/3}|a|)\ge f^{-1/3}(f^{1/3} \alpha/f -C) = \alpha/f - f^{-1/3}C\ge \alpha/2f$$
for small $f$.   Thus $||g_{\epsilon, z}1_{\{|w| \le C\}}||_2 \le k_{\epsilon} e^{-\epsilon \alpha/4f}$. Therefore in the following we assume $|w|$ is large so we can use asymptotic expansions of the Airy function.   First assume that $|\arg(a+ib)| < \pi - \delta_1$ for some small $\delta_1 > 0$.  We write $1_{\delta_1}$ for the indicator function of this set.  Then we can use the asymptotic behavior of $\Ai(w)$,  $|\Ai(w)| \le C|w|^{-1/4} e^{-\text{Re}\zeta}$, where $\zeta = 2w^{3/2}/3$.  If $a\ge 0$ 
$$-\text{Re}w^{3/2} \le f^{1/2} (\sqrt{a^2 + b^2} +a)^{1/2} (\sqrt{a^2 + b^2} -a -a)/\sqrt 2 \le f^{1/2} (\sqrt{a^2 + b^2}^2 - a^2)(\sqrt{a^2 + b^2} -a)^{1/2}/\sqrt 2$$ 
$$= f^{1/2}|b|(\sqrt{a^2 + b^2}-a)^{1/2}/\sqrt 2 \le f^{1/2}|b|^{3/2}/\sqrt 2 = |\beta|^{3/2}/\sqrt 2 f.$$
Thus $$\epsilon|x|/2 + (2/3) \text{Re}w^{3/2}  \ge \epsilon (a + \alpha/f)/2 - \sqrt 2 |\beta|^{3/2}/3f.$$  If we demand that $\alpha \ge 4 |\beta|$ and $\epsilon \sqrt {\alpha} \ge |\beta| $ then $$\epsilon|x|/2 + (2/3) \text{Re}w^{3/2}  \ge \epsilon \alpha/4f + \epsilon \alpha/4f - \sqrt 2 /6f \alpha^{1/2}|\beta| \ge \epsilon \alpha/4f.$$
Thus we have $||g_{\epsilon, z}1_{\{|w| \ge C, a \ge 0\}}1_{\delta_1}||_2 \le k_{\epsilon} e^{-\epsilon \alpha/4f}$.
\vskip .1in
Suppose now that $a\le 0$.  Then $$-\text{Re}(a+ ib)^{3/2} =  (\sqrt{a^2 + b^2} -|a|)^{1/2}(\sqrt{a^2 + b^2} +2|a|)/\sqrt 2\le \sqrt 2|b| (\sqrt{a^2 + b^2} +|a|)^{1/2}$$
$$\le 2|b|(a^2 + b^2)^{1/4}\le 2|b|^{3/2} + 2|b||a|^{1/2}$$ and thus since $|x| = |a| + \alpha/f$
$$\epsilon |x| + 2\text{Re}w^{3/2}/3 \ge \epsilon|x|/2 + \epsilon \alpha/4f + \epsilon\alpha/4f - 4|b|^{3/2}f^{1/2}/3  + \epsilon|a|/2 - 4|b| |a|^{1/2}f^{1/2}/3.$$
Minimizing the last two terms over $|a|$ we find
$$\epsilon |x| + 2\text{Re}w^{3/2}/3 \ge (\epsilon|x|/2 + \epsilon \alpha/4f) +  \epsilon\alpha/4f - 4|\beta|^{3/2}/3f - 8\beta^2/9\epsilon f$$ 
We now require $\alpha \ge 8|\beta|$ and $\epsilon \sqrt{\alpha} \ge 8|\beta|$.  Then one can check that the last three terms add to something positive.  Thus $||g_{\epsilon, z}1_{\{|w| \ge C, a \le 0\}}1_{\delta_1}||_2 \le k'_{\epsilon}e^{-\epsilon \alpha/4f}$. 
\vskip .1in
Finally, consider the region where $\pi \ge |\arg(a+ib)| \ge \pi - \delta_1$ where $\delta_1$ is small.  Then we can use the asymptotics $|\Ai(-u)| \le C|u|^{-1/4} e^{|\text{Im}2u^{3/2}/3|}$ where $-u = w = f^{1/3}(a+ib)$.  Thus $u = f^{1/3}(|a| - ib)$ and from above $|\text{Im}u^{3/2}| = ( \sqrt{a^2 + b^2} - |a|)^{1/2}( \sqrt{a^2 + b^2} + 2|a|)/\sqrt2$.  But this has just been estimated and we thus find   $||g_{\epsilon, z}1_{\{|w| \ge C \}}(1-1_{\delta_1})||_2 \le k_{\epsilon} e^{-\epsilon \alpha/4f}$.  This proves the lemma.
\end{proof}

\begin{lemma}\label{analyticcont} \label{f_1}
For $f> 0$ the weighted resolvent $\zeta_1(x)(H_{f,0} - z)^{-1}\zeta_2(x)$ has an analytic continuation from the upper half plane to $\mathbb{C}$ as a bounded operator as long as $|\zeta_j(x)| \le C_{\gamma}e^{-\gamma |x|^{1/2}}$ for all $\gamma > 0$.  If $\text{Im} f >0$, $(H_{f,0} - z)^{-1}$ is analytic in $f$ and entire in $z$ and if $f_1>0$ then $$\lim_{f\to f_1, \text{Im}f \downarrow 0}|| \zeta_1(x)\Big((H_{f,0} - z)^{-1} - (H_{f_1,0} - z)_c^{-1}\Big) \zeta_2(x)|| = 0$$ where the subscript $c$ indicates the analytic continuation of $\zeta_1(x)(H_{f_1,0} - z)^{-1}\zeta_2(x)$ from the upper half $z$ plane.
\end{lemma}
\begin{proof}
Since the Airy functions are entire functions of their arguments, the lemma just follows from their asymptotic behavior (see \cite{AS}).   
\end{proof}
Note that (\ref{resolventeqn}) then shows that if $f> 0$, $\zeta_1(x)(H_f - z)^{-1}\zeta_2(x)$ has a meromorphic continuation to $\mathbb{C}$.
\begin{proposition}\label{intresolventconvergence}
Suppose $\lambda_0 < 0$.  Choose $r$ with $0 < r < \text{dist}(\lambda_0, \sigma(H)\setminus \{\lambda_0\})$ and $r < |\lambda_0|/9$.  Then if $\epsilon > \sqrt{8|\lambda_0|/9}$,  $$||e^{-\epsilon |x|} ((H_f - z)_c^{-1}  - (H-z)^{-1})e^{-\epsilon |x|}|| \to 0$$ as $f\downarrow 0$ uniformly for $|z-\lambda_0| = r$.  Here $e^{-\epsilon |x|} (H_f - z)_c^{-1}e^{-\epsilon |x|}$ is the meromorphic continuation of the resolvent (with weights) from the upper half plane to a neighborhood of the point $\lambda_0$.  
\end{proposition}

\begin{proof}
We use Lemma \ref{resolventconvergence} to show the convergence of the resolvent for $\text{Im} z \ne 0$, and even for the limits onto the real axis from above and below.  To see the convergence below the real axis in the region stated we use a version of (\ref{K_c}) without the weights $V_1$ and $V_2$, along with Lemma \ref{jumpconvergence}.  This shows that $$||e^{-\epsilon |x|} ((H_{0,f} - z)_c^{-1}  - (p^2-z)^{-1})e^{-\epsilon |x|}|| \to 0$$ as $f\downarrow 0$ uniformly for $|z-\lambda_0| = r$. In particular it follows that $(I + V_1 (H_{f,0} -z)^{-1}_cV_2)^{-1} \to (I + V_1 (p^2 -z)^{-1}V_2)^{-1}$ as $ f \downarrow 0$.  Then we use (\ref{resolventeqn}) and Lemma \ref{analyticcont} to show that $(H_f -z)^{-1}$ (with the weights given in that lemma) has a meromorphic continuation from the upper half plane to $\mathbb{C}$.  Finally with the weights $e^{-\epsilon|x|}$, we can see using (\ref{resolventeqn}) for $(H_f - z)_c^{-1}$, we have the convergence stated in this proposition.
\end{proof}

\section{Resonances near the negative real axis}\label{rnnra}

\begin{theorem} \label{asymptotic expansion}
Suppose $\lambda_0$ is an eigenvalue of $H$. Then for small enough $f>0$ there is a resonance $\lambda(f)$ of $H_f$ near $\lambda_0$.  $\lambda(f)$ has an asymptotic expansion in (non-negative) powers of $f$.  The terms of this expansion can be calculated as if $\lambda(f)$ were a real eigenvalue of $H_f$ using the Rayleigh-Schr\"odinger perturbation expansion.  Thus the expansion coefficients are real although $\lambda(f)$ is not.
\end{theorem}

\begin{proof}
Let $\psi_0$ be the normalized eigenfunction of $H$ corresponding to $\lambda_0$.  Since $V$ has compact support $\psi_0 = ae^{\sqrt{-\lambda_0}x} + be^{-\sqrt{-\lambda_0}x}$ and thus it is clear that $\lambda_0 < 0$. We see that in fact  $\psi_0 = a_{\pm} e^{-\sqrt{-\lambda_0}|x|}$ for large $|x|$  with $\pm x > 0$.   Because of the decay rate of $\psi_0$ we learn that for $r$ as in Proposition \ref{intresolventconvergence},   $$\lim_{f\downarrow 0}-(2\pi i)^{-1}\oint_{|z-\lambda_0| = r} (\psi_0, (H_f- z)_c^{-1}\psi_0)dz = -(2\pi i)^{-1}\oint_{|z-\lambda_0| = r} (\psi_0, (H- z)^{-1}\psi_0) dz= 1.$$  As we can make $r > 0$ as small as we like we see that $H_f$ has a resonance $\lambda(f) \to \lambda_0$ as $f\downarrow 0$. The pole in $(H-z)^{-1}$ at $\lambda_0$ corresponds to a simple zero of the Wronskian and similarly from the convergence of projections $\lambda(f)$ is a simple pole in the resolvent $(H_f-z)^{-1}$ for $\text{Im}f >0$.  According to the convergence result of Lemma \ref{f_1}  for small $f>0$ the pole of $(\psi_0, (H_f- z)_c^{-1}\psi_0)$ is simple.  Thus for small enough $r$ and $f>0$,$\oint_{|z-\lambda_0| = r} (\psi_0,(z-\lambda(f)) (H_f- z)_c^{-1}\psi_0)dz = 0$.  It follows (again for small enough $r$ and $f$) that 

\begin{equation}\label{resonanceformula}
\lambda(f) = \frac{ \oint_{|z-\lambda_0| = r} (\psi_0, z(H_f- z)_c^{-1}\psi_0)dz }{\oint_{|z-\lambda_0| = r} (\psi_0, (H_f- z)_c^{-1}\psi_0)dz}.
\end{equation}
 We now substitute the asymptotic expansion of $(H_{0,f} - z)_c^{-1}$ from Lemma \ref{asymptoticexpansion} into (\ref{resolventeqn}) (note that the rank one term as estimated in Lemma \ref{jumpconvergence} has a zero asymptotic expansion and can thus be ignored).  Thus $\lambda(f)$ has an asymptotic expansion in powers of $f$.  If we instead did the same thing with $fx$ replaced by a real bounded function $fW$ with $\lim_{|x| \to \infty} W(x) = 0$ we would again get an asymptotic expansion, this time for the eigenvalue of $H + fW$,  which is actually convergent for small $f$ and whose terms are real.  The terms of this expansion match up exactly with those of our expansion of $\lambda(f)$ except of course with $x$ replaced by $W$.  Thus the asymptotic expansion of $\lambda(f)$ is exactly given by Rayleigh - Schr\"odinger perturbation theory for the non-existent eigenvalue of $H_f$.  It is non-existent because no linear combination of the two linearly independent Airy functions is square integrable near $-\infty$.  Thus $H_f$ has no eigenvalues.

\end{proof}

This result has been known for many years (\cite{GGLM, GG1, H, HaS}) in the dilation analytic framework.  In the latter framework $\lambda(f)$ is an actual eigenvalue of the dilated Hamiltonian unlike in our framework.  But of course the dilation analytic framework requires that the potential, $V$, be dilation analytic in some angle.  The result was also proved by Graffi and Grecchi (\cite{GG}) for Hydrogen in an electric field using the separability in squared parabolic coordinates.
\vspace{.1 in}

As is well known, for $\text{Im} f \ne 0$ the resolvent $(H_{0,f} - z)^{-1}$ is analytic in $f$ and entire in $z$.  For such $f$ the operator $H_{0,f}$ has domain equal to $D(p^2)\cap D(x)$.  Let $N_f = \{ f t + s^2 : t \ \text{and}  \ s \in \mathbb{R} \}$.  For all non-real $f$ the spectrum is empty so the resolvent is of course bounded.  But in addition we have the explicit bound $||(H_{0,f} - z)^{-1}|| \le 1/\text{dist}(z,N_f)$ when $z$ is outside the closure of the numerical range, $N_f$, of $H_{0,f}$.  Of course this bound also holds if $f=0$.  

\begin{proposition}\label{convergencefcomplex}
Consider the Hamiltonian $H_f$ for $f$ non-real.  If $f$ is not real the resolvent of $H_f$ is meromorphic in $\mathbb{C}$.  Suppose $C = \{z: |z-z_0| = r\}$ is a circle in $\mathbb{C}$ which is disjoint from the spectrum of $H$ and from $N_f$.  Then for small enough $f$, $C$ is disjoint from the spectrum of $H_f$ and allowing $f \to 0$ in such a way that the distance between the circle $C$ and $N_f$ is bounded away from $0$ we have the convergence of projections $$ \lim_{f \to 0}|| (2\pi i)^{-1}\int_{|z-z_0| = r}(H_f-z)^{-1}dz  - (2\pi i)^{-1}\int_{|z-z_0| = r}(H-z)^{-1}dz|| =0.$$ 
\end{proposition}

\begin{proof}
We use the formula (\ref{resolventeqn}): $$(H_f-z)^{-1} = (H_{f,0} -z)^{-1} - (H_{f,0} -z)^{-1}V_2(1+  V_1 (H_{f,0} -z)^{-1}V_2)^{-1}V_1 (H_{f,0} -z)^{-1}.$$
The convergence of $V_1 (H_{f,0} -z)^{-1}V_2$ to $V_1 (p^2 -z)^{-1}V_2$ in norm is easy  to show.  Equation (\ref{resolventeqn}) then implies that for small $f$,$(H_f-z)^{-1}$ has no spectrum on the circle and we can integrate $(H_f-z)^{-1} - (H-z)^{-1}$ around $C$.  We use  (\ref{resolventeqn}) for both resolvents and note that the term $(H_{f,0} -z)^{-1} - (p^2 -z)^{-1}$ occurring in the difference of resolvents integrates to zero.  Thus consider $(H_{f,0} -z)^{-1}V_2 - (p^2 -z)^{-1}V_2$ for example.  We have $||\Big((p^2 + fx -z)^{-1} - (p^2-z)^{-1}\Big)V_2|| \le c|f|$ where we are using the fact that the distance between $C$ and $N_f$ is bounded away from zero.  The term $V_1(H_{f,0} -z)^{-1}$ is treated in the same way.  This completes the proof.
\end{proof}

\begin{corollary}
From proposition \ref{convergencefcomplex} and the formula (\ref{resonanceformula}) we see that for $f >0$ and small, $\lambda(f)$ is the boundary value from $\text{Im} f>0$ of a function analytic in $f$.  The analytic continuation of $\lambda(f)$ to $\text{Im}f >0$ with $|f| >0$ but small, is a simple eigenvalue of $H_f$.  
\end{corollary}

\begin{theorem}
Suppose $\lambda_0$ is a negative eigenvalue of $H = p^2 +V$ where $V$ is a real bounded measurable function of compact support.  Suppose $\Phi$ is the normalized eigenfunction corresponding to $\lambda_0$.  Then for small positive $f$ the Hamiltonian $H_f = H + fx$ has a resonance $\lambda(f)$ near $\lambda_0$ with real part given as in Theorem \ref{asymptotic expansion} and imaginary part 
$$-\text{Im}\lambda(f) =  \frac{1}{4\sqrt{-\lambda_0}} e^{-\frac{4}{3f}(-\lambda_0 -f\lambda_1)^{3/2}}(Ve^{-\sqrt{-\lambda_0}x},\Phi)^2(1+O(f))$$ where $\lambda_1 = (\Phi,x\Phi)$. 
We have $(Ve^{-\sqrt{-\lambda_0}x},\Phi) = -2\kappa_- \sqrt{-\lambda_0}$ where $\Phi(x) = \kappa_-e^{\sqrt{-\lambda_0} x}$ for $x$ near $-\infty$.
\end{theorem}

\begin{proof}
We follow Howland in  \cite {HJ} where he computes an exponentially small imaginary part of a resonance caused by a barrier which is becoming infinite in extent.  Our situation is similar but a bit more complicated.  Nevertheless much of the analysis below is lifted from Howland's paper.  Let $Q^+(f,z) = V_1(H_{f,0} -z)^{-1}_cV_2$  where the subscript $c$ and the superscript $+$ indicate analytic continuation from the upper half plane.  Similarly a $-$ superscript will indicate analytic continuation from the lower half plane.  We let $L(f,z) = (Q^+(f,z) + Q^-(f,z))/2$ and $D(f,z) =  (Q^+(f,z) - Q^-(f,z))/2$.  It follows from (\ref{K_c}) that $D(f,z)$ is a rank one operator with kernel 
\begin{align} \label{rank1}
&D(f,z)(x,y) = \pi i f^{-1/3} g_1(x;f,z)g_2(y;f,z) \  \text{with} \nonumber  \\
& g_1(x;f,z) = V_1(x)\Ai(f^{1/3}(x-z/f)) \ \text{and}  \nonumber \\
& g_2(x;f,z) = V_2(x)\Ai(f^{1/3}(x-z/f)).  
\end{align}
We saw that  $I + Q^+(f,z )$ has a one dimensional kernel at the resonance $ z = \lambda(f)$ converging to $\lambda_0$ as $f \downarrow 0$. Similarly if we consider $I + L(f,z)$ the same proof shows that this operator has a simple eigenvalue $-1$ at $z = \mu(f)$ with $\mu(f)$ converging to $\lambda_0$ as $f \downarrow 0$.  We assume that $|V_1| + |V_2| =0$ exactly where $V=0$ and $V_1/V_2$ is bounded above and below on the set where $|V| >0$.  We set $g=V_1/V_2 $ on the set $\{|V| > 0\}$ and $=1$ on the complement. It follows that $g^{-1}L(f,\mu(f))g = L(f, \overline {\mu(f)})^*$.  Since $\sigma(L(f,\overline {\mu(f)})) =\overline{ \sigma(L(f,\overline {\mu(f)})^*) }$, $L(f,z)$ has eigenvalue $-1$ for both $z = \mu(f)$ and $z = \overline {\mu(f)}$.  But there is only one such $z$ near $\lambda_0$ for small $f > 0$ and thus $\mu(f)$ is real.  Let $J(f) = (2\pi i)^{-1}\int _{|z+1| = \epsilon}(z-L(f,\mu(f))^*)^{-1}dz$ for small $\epsilon > 0$.  Since in the limit $f\downarrow 0$, $J(f) \to J(0)$ which projects onto $V_2\Phi$, where $(H-\lambda_0)\Phi =0$, $\phi_f = J(f)V_2\Phi$ satisfies $(L(f,\mu(f))^* + I)\phi_f =0$ for small $f>0$.  Similarly define $\tilde J(f)  = (2\pi i)^{-1}\int _{|z+1| = \epsilon}(z-Q^+(f,\lambda(f)))^{-1}dz$ for small $\epsilon > 0$.  Then  $\psi_f =\tilde J(f) V_1\Phi$ is in the kernel of $I + Q^+(f,\lambda(f)$ for small $f >0$.

\begin{align*}
0&= (\phi_f, (I + Q^+(f,\lambda(f))\psi_f)  \\
& = (\phi_f,(I + L(f,\mu(f)) + L'(f,\mu(f))(\lambda(f) - \mu(f)) + D(f,\lambda(f)) )\psi_f)  + O(|\lambda(f) - \mu(f)|^2) \\
& = (\phi_f, L'(f,\mu(f)) \psi_f) (\lambda(f) - \mu(f)) + (\phi_f,D(f,\lambda(f))\psi_f) + O(|\lambda(f) - \mu(f)|^2) 
\end{align*}

Here we have used the fact that  $L(f,z)$ is analytic in $z$ with derivatives bounded for $z$ in a neighborhood of $\lambda_0$ as $f \downarrow 0$.  We have $(\phi_f, L'(f,\mu(f))\psi_f) \to (V_2\Phi, V_1(p^2 - \lambda_0)^{-2}V_2 V_1\Phi) = ||\Phi||^2$ and thus

$$\lambda(f) - \mu(f) = - (\phi_f,D(f,\lambda(f))\psi_f)/(\phi_f, L'(f,\mu(f)) \psi_f) + O(|\lambda(f) - \mu(f)|^2) $$ 

Consider $D(f,\lambda(f))$.   We have $$f^{-1/6} g_j(x) =  V_j(x)\frac {1}{2\sqrt{\pi}(-\lambda(f))^{1/4}}e^{-\frac{2 (-\lambda(f))^{3/2}}{3f}}e^{-x\sqrt{-\lambda(f)}} (1+ O(f)) $$ so
\begin{align*}
&(\phi_f,D(f,\lambda(f))\psi_f) =  \\
&(i/4\sqrt{-\lambda(f)})e^{-\frac{4 (-\lambda(f))^{3/2}}{3f}} \int \overline {\phi_f(x)} V_1(x) e^{-\sqrt{-\lambda(f)} x}(1+O(f))  e^{-\sqrt{-\lambda(f)} y } V_2(y)\psi_f (y)dxdy
\end{align*}
From above we have $\text{Im}\lambda(f) =O(f^n)$ for all $n$ and $\lambda(f)$ has an asymptotic expansion in $f$ given by the Rayleigh-Schr\"odinger series.  Thus $\lambda(f) = \lambda_0 + f (\Phi,x\Phi) + O(f^2)$ (here we normalize $||\Phi|| = 1$.)
Thus $$ \lambda(f) - \mu(f) = \frac{-i}{4\sqrt{-\lambda_0}} e^{-\frac{4 (-\lambda(f))^{3/2}}{3f}} \Big ( (Ve^{-\sqrt{-\lambda_0}x},\Phi)^2 + O(f) \Big )$$
We can compute $(Ve^{-\sqrt{-\lambda_0}x},\Phi) = - \int [(p^2 - \lambda_0)\Phi(x)]e^{-\sqrt{-\lambda_0}x} dx$ using integration by parts.  We find  $(Ve^{-\sqrt{-\lambda_0}x},\Phi) = -2\kappa_-\sqrt{-\lambda_0}$,where $\Phi(x) = \kappa_- e^{\sqrt{-\lambda_0} x}$ for $x$ near $-\infty$.
Then  $$ \lambda(f) - \mu(f) = -i\sqrt{-\lambda_0} \kappa_-^2e^{-\frac{4 (-\lambda_0 - f(\Phi,x\Phi))^{3/2}}{3f}}(1 + O(f)).$$
Since $\mu(f)$ is real we learn that 
\begin{align}
&\text{Im}\lambda(f) = -\sqrt{-\lambda_0} \kappa_-^2e^{-\frac{4 (-\lambda_0 - f(\Phi,x\Phi))^{3/2}}{3f}}(1 + O(f)). \\  \nonumber
\end{align}
\end{proof}
Note that in the above proof $\mu(f) = \text{Re}\lambda(f)$ up to a function with zero asymptotic expansion.

\section{Resonances near the positive real axis}\label{pra}
Let $$A_1(x) = \Ai(e^{2\pi i/3}f^{1/3}(x-z/f))$$
$$A_2(x) = \Ai(f^{1/3}(x-z/f))$$
For fixed $f>0$ and $\text{Im}z>0$, $A_1(x) \in L^2(-\infty, 0)$ while if $\text{Im}z<0$ $A_1(x) \notin L^2(-\infty,0)$. On the other hand, for fixed $f>0$,  $A_2(x) \in L^2(0,\infty)$ for all $z$.  
We know that if $z$ is a resonance then there is a solution of Schr\"odinger's equation $- \psi'' + (V(x) +fx) \psi = z\psi$ which has the form $cA_2(x)$ for large positive $x$ and $c'A_1(x)$ for large negative $x$.  We note the following obvious but important fact: If we are given $(\psi (a), \psi'(a)) = (\alpha,\beta)$ , then for any $b \in \mathbb{R} $, both $\psi(b)$ and $\psi'(b)$ are entire functions of the variables $(z,f,\alpha,\beta)$.  

Suppose $z$ is near the positive real axis and the support of the potential $V$ is contained in $(-L,L)$.  Choosing $\psi(-L) = 1$ and $\psi'(-L) = A_1'(-L)/A_1(-L)$ we note that $\psi(L)$ and $\psi'(L)$ are analytic functions of $f,z,$ and $\psi'(-L)$.  It is easily seen from \cite{AS} that for $\arg z \in (-2\pi/3 + \delta, 0]$ and $k = \sqrt{z}$ with $\arg(k)\in(-\pi/2,0]$.
$$A_1'(-L)/A_1(-L) =: -i\lambda(k,f) = -ik(1+O(f))$$
uniformly for $|k| > \delta > 0$.  On the other hand from \cite{AS}
\begin{align} 
A_2'(L) = -\pi^{-1/2} f^{1/6} \sqrt{k}(1+O(f)) \Big((1+O(f^2))\cos(\zeta + \pi/4) + O(f)\sin(\zeta + \pi/4)\Big) \label{A'2}\\
A_2(L) = \pi^{-1/2} f^{1/6} k^{-1/2}(1+O(f))\Big((1+O(f^2))\sin(\zeta + \pi/4) + O(f)\cos(\zeta + \pi/4)\Big) \label{A2} 
\end{align}
with $$\zeta = \frac{2k^3}{3f}(1-fL/k^2)^{3/2}$$
To get an idea what we are dealing with we look at the leading order as $f\downarrow 0$ of $(\psi(L), \psi'(L))$ which arises by propagating from $-L$ to $L$ with $(\psi(-L), \psi'(-L)) = (1, -ik)$
and $f=0$.  Thus we are solving the Schr\"odinger equation $-\psi'' + V(x)\psi = k^2\psi$ with $\psi(x) = e^{-ikx}$ for $x$ to the left of the support of $V$.  Then to this leading order $\psi'(L)/\psi(L) = g(k)$ is analytic in $k$ unless $\psi(L) = 0$.  To leading order in the sense that we neglect $O(f)$ terms in (\ref{A'2}) and (\ref{A2}) we obtain
$$A_2'(L)/A_2(L)= -k /\tan(\zeta + \pi/4),$$ and the equation $A_2'(L)/A_2(L)=\psi'(L)/\psi(L)$ which holds when $k$ is a resonance takes the form

$$e^{-2i\zeta} =i\frac{\psi'(L) + ik \psi(L)}{\psi'(L) - ik\psi(L)}.$$

Let us use unitarity of the S-matrix to get a relation between $\psi'(L)$ and $\psi(L)$ for real $k$ given the initial $(\psi(-L), \psi'(-L)) = (1, -ik)$.  Here we assume $f=0$.  We have $\psi(x) = e^{-ikx}$ for $x$ to the left of $V$ and $c_1e^{ikx} + c_2e^{-ikx}$ to the right of $V$.  Unitarity gives $|c_2|^2 = 1 + |c_1|^2$.  (Notice that $c_1 = r(-k)/t(-k), c_2 = 1/t(-k)$ where $t(k)$ is the transmission amplitude for a particle of momentum $k$  and $r(k)$ is the reflection amplitude for this momentum.)  We compute $$-\frac{c_1e^{ikL}}{c_2e^{-ikL}} = \frac{\psi'(L) +ik\psi(L)}{\psi'(L)-ik\psi(L)}= -r(-k)e^{2ikL}.$$  Here $r(-k)$ is the analytic continuation of the reflection amplitude from $k>0$ to $\text{Im}k <0$. It follows that 

$$\left|\frac{\psi'(L) +ik\psi(L)}{\psi'(L)-ik\psi(L)}\right|^2 = \frac{|c_1|^2}{1 + |c_1|^2} < 1$$ 
for $k$ real.  And we see that for real $k$, $\frac{\psi'(L) +ik\psi(L)}{\psi'(L)-ik\psi(L)} = 0$ exactly when the (right) reflection coefficient, $r(-k)$, is zero. 

Define 
      \[S(k) =           \begin{pmatrix} t(k) & r(-k) \\ 
                                         r(k) & t(-k) 
                         \end{pmatrix}\] 
where  the transmission and reflection amplitudes are given by          
\begin{align}
&t(k) = 1-2\pi i |2k|^{-1} \langle k|T(k^2+i0)|k\rangle \\
&r(k) = -2\pi i |2k|^{-1}  \langle -k|T(k^2+i0)|k\rangle.
\end{align}
Then for real $k$ these quantities satisfy the usual unitarity relations
$$|r(k)|^2 + |t(k)|^2 = 1$$
$$\overline{r(-k)}t(k) + r(k)\overline{t(-k)} = 0$$
or in matrix form

   $$S(k)^*S(k) = I$$
 This is a consequence of the unitarity of the $S-\text{matrix}$ as given in (\ref{STdef}).  Note that it follows from the symmetry of the resolvent kernel $(H-z)^{-1}(x,y)$ that $t(k) = t(-k)$ and thus from the unitarity relations that $|r(k)| = |r(-k)|$.

Going back to $f>0$, we see that the functions $\psi'(L) \pm ik \psi(L)$ are of the form $$h_{\pm}(k,f,\lambda(k,f))$$ where the  functions $h_{\pm}$ are analytic in their arguments. If $\arg k \in (-\pi/3 + \delta, 0)$ we have $$\psi'(L) \pm ik \psi(L) = h_{\pm}(k,0,k) + O(f).$$    
Consider a real point $k_0> 0$ for which the (right) reflection coefficient for scattering in potential $V$ is non-zero.  In a neighborhood of $k_0$, say $|k-k_0| \le  \epsilon$,
$$\frac{\psi'(L) +ik\psi(L)}{\psi'(L)-ik\psi(L)} = G(k)e^{2ikL}(1+b(k,f))$$ where $G$ is analytic in this neighborhood and non-zero while for real $k$, $|G(k)| < 1$. The function $b(k,f)$ is $C^1$ in $k$ with $\partial b(k,f)/\partial k$ bounded for $\delta^{-1} > |k| > \delta$ ($\delta > 0$), $k$ near the positive real axis, and $f>0$ small. (See Appendix \ref{AppendixAiry2}).  It follows that  a resonance $k$ of $H_f$ in this neighborhood obeys 

\begin{equation} \label{errornopole}
G(k) = -ie^{-4ik^3/3f}(1+a(k,f)).
\end{equation}
where we have cancelled out the factor $e^{2ikL}$ using  
$2\zeta = \frac{4k^3}{3f}(1-fL/k^2)^{3/2} = \frac{4k^3}{3f} - 2Lk +O(f)$.  It follows that for some integer $j$ and some branch of the logarithm

$$ k^3 = (3if/4)\log(iG(k)) + 3\pi fj/2   -(3fi/4)\log(1+a(k,f)) = (3if/4)\log(iG(k)) + 3\pi fj/2 + fc(k,f)$$ 
where $c(k,f) = O(f)$.  Without loss of generality we choose a continuous branch of the logarithm with $\arg( \log iG(k_0)) \in (-\pi,\pi]$.  For $f$ small we will choose $j$ very large so that $$\eta_0(j): = (3\pi f j/2)^{1/3}$$ is not far from $k_0$.  We require $|\eta_0(j) - k_0| < \epsilon < k_0/2$.
We thus obtain $$ k^3 - \eta_0^3 = (3if/4)\log(iG(k)) + fc(k,f)$$ 
or 
$$k = \eta_0\Big( 1 + (3if/4\eta_0^3)\log(iG(k) + \eta_0^{-3}fc(k,f)\Big)^{1/3}$$
A simple contraction mapping argument shows that for $|k-k_0| < \epsilon$ with $\epsilon$ sufficiently small, there is a unique solution to this equation which we call $k(j)$.  This string of resonances  satisfies 
$$k(j) = \eta_0(j) + i(f/4\eta_0(j)^2)\log G(i\eta_0(j)) + O(f^2).$$

We remark on the nature of the solution:
Notice that there are order of magnitude $f^{-1}$ integers $j$ with $|\eta_0(j) - k_0| < \epsilon$. (More exactly since $\Delta \eta_0 = \epsilon$,  to order $\epsilon^2$, $\Delta \eta_0(j)^3 = 3\eta_0^2 \epsilon$, so there are $\Delta j = (2\eta_0^2/\pi f)\epsilon = (2k_0^2/\pi f)\epsilon $ such integers $j$ to order $\epsilon^2$.) Write $$\log (iG(k)) = -l(k) + i\theta(k)$$  where for $\epsilon$ small enough, $l(k) > 0$. Then we get $\text{Im}k(j) = -fl(\eta_0(j))/4\eta_0(j)^2 + O(f^2)$ and  $\text{Re}k(j) = \eta_0(j)  - \theta (\eta_0(j))/4\eta_0(j)^2)f + O(f^2)$. We have thus shown

\begin{theorem}
Suppose the reflection coefficient is non- zero at $k_0>0$.  Then if $\epsilon > 0$ is small enough, $H_f$ has resonances in the disk $|k-k_0| < \epsilon$.  If $k$ is a such a resonance there exists an integer $j$ such that $k$  is one of the resonances $k(j)$ found above which in particular satisfy
$$ k(j) = \eta_0(j) - 4^{-1}\Big(\theta(\eta_0(j))/\eta_0(j)^2 + il(\eta_0(j))/\eta_0(j)^2\Big)f + O(f^2).$$ Here $l(k) >0$ for $k$ near $k_0$ and $j$ is allowed to vary in an interval so that $|k-k_0|< \epsilon$. The linear density of resonances along the positive real axis near a point $k_0$ where the reflection coefficient is non-zero is to leading order $2k_0^2/\pi f$.

\end{theorem}

We now consider how to calculate resonances in the neighborhood of a point $k_0 > 0$ where the reflection coefficient vanishes.  The functions $h_{\pm} = h_{\pm}(k,f,\lambda)$ are analytic functions of their three arguments in a neighborhood of $(k_0,0,k_0)$ but $\lambda(k,f)$ is not analytic in $f$, rather analytic in a product set of the form $\{k \in \mathbb{C}: |k-k_0| < \epsilon \} \times \{f: 0<|f| < \epsilon, |\arg f| < \epsilon\}$ and $C^{\infty}$ in $\{k \in \mathbb{C}: |k-k_0| < \epsilon \} \times \{f: |f| < \epsilon, |\arg f| < \epsilon\}$ in the sense that the derivatives are continuous in $f$ up to $f=0$ (see  Appendix \ref{AppendixAiry2}).  Let $\mu(k,f) = \lambda(k,f) - k$ (thus $\mu(k,0) = 0$). Let us write 
$$\frac{\psi'(L)}{-ik\psi(L)} = \frac{A_2'(L)}{-ikA(L)} = \frac{1+ a(k,f)) e^{2i(\zeta+\pi/4)} + 1+b(k,f)}{(1+c(k,f) e^{2i(\zeta +\pi/4)}- (1+d(k,f))}$$
where according to Appendix \ref{AppendixAiry2} the functions $a,b,c,d$ are $C^{\infty}$ in $f$ for $f\ge 0$ and small and analytic in $k$ for $k$ near a point on the positive real axis.  In addition they are all $0$ when $f = 0$. Inverting the linear fractional transformation we have 

$$ie^{2i\zeta} = \frac{\psi'(L) -ik\psi(L) + d\psi'(L) - ik b\psi(L)}{\psi'(L) + ik \psi(L) +c\psi'(L) + ik a\psi(L)}$$
We are interested in resonances near a point $k_0 > 0$ where the reflection coefficient vanishes.  This means $\psi'(L) + ik \psi(L) = 0$ when $f =0$.  We have $\psi'(L) + ik \psi(L) = h_+(k,f,\lambda(k,f))$ with $h_+$ analytic in its three arguments and 0 at $(k_0,0, \lambda(k_0,0)) $.  The quantity $\tilde{ a}: = c\psi'(L) + ik a\psi(L)$ is $C^\infty$ in small $f\ge 0$ and analytic in $k$ near $k_0$.  Of course $a(k,0) = 0$.  Thus the denominator is an analytic function of $k,f, \mu,$ and $\tilde{a}$ with a zero at $(k_0,0,0,0)$.  We can thus use the Weierstrass preparation theorem to write 

$$e^{2i\zeta} = g(k,f)/\text{prep}(k,f)$$
where 
$$\text{prep}(k,f) = (k-k_0)^p + b_{p-1}(f,\mu,\tilde a)(k-k_0)^{p-1} + \cdots + b_0(f,\mu,\tilde a)$$
Here $p$ is the order of the zero of $h_+(k,0,k)$. The function $g$ is analytic in $k$ near $k_0$ and $C^\infty$ in small $f\ge0$. We write $b_s(f,\mu,\tilde a) = b_s(f,k)$.  $b_s(f,k)$ is analytic in $k$ near $k_0$ and $C^\infty$ in small $f\ge 0$.  We have $b_s(0,k) = 0$.  Defining $m(k,f) = g(k,f)e^{2ik^3/3f[(1-fL/k^2)^{3/2}-1]}$, we need to solve

\begin{equation} \label{peqn}
e^{4ik^3/3f} = \frac{m(k,f)}{\text{prep}(k,f)}.
\end{equation}
We proceed by iteration.  Define $k_1(j)$ and $\lambda_f(j)$ by the equations
\begin{equation} \label{k_1}
k_1(j) = \lambda_f(j) + k_0 = (3\pi f j/2)^{1//3}. 
\end{equation}
Here $j$ is chosen very large for $f$ small so that $k_0/2 > \delta_0 > |\lambda_f| \ge 2f^{(1-\epsilon)/p}$ where $\epsilon \in (0,1)$ and $\delta_0$ is small. With this lower bound we see that for small $f$, $|\text{prep}(k_1,f)| \ge 2f^{1-\epsilon} + O(f)$.  The equation we want to solve is 
$$k^3 = 3\pi f j/2 -(3if/4)(\log m(k,f) - \log \text{prep}(k,f)).$$ 
We set 
$$k_n^3 = k_1^3 - (3if/4)(\log m(k_{n-1},f) - \log \text{prep}(k_{n-1},f)) ; n \ge 2 $$
where we take the cube root closest to the positive real axis. Let us assume $|k-k_0| \ge f^{(1-\epsilon)/p}$ and estimate $\text{prep}(k,f)$ and its derivative with respect to $k$.  We have 
$$\text{prep}(k,f) = (k-k_0)^p + \sum_{j=1}^p\int_0^f (\partial{b}_{p-j}(s,k)/\partial s) (k-k_0)^{p-j} ds.$$  Thus
$$|\text{prep}(k,f)| \ge |(k-k_0)^p|( 1 - \sum_{j=1}^p c_jf|k-k_0|^{-j})\ge |k-k_0|^p(1-cf^{\epsilon}).$$ 

\vspace{.5cm}

$$\partial\text{prep}(k,f)/\partial k = p(k-k_0)^{p-1} + \sum_{j=1}^p\int_0^f (\partial^2{b}_{p-j}(s,k)/\partial s \partial k) (k-k_0)^{p-j}ds $$ $$ + \sum_{j=1}^{p-1}\int_0^f(\partial {b}_{p-j}(s,k)/\partial s) (p-j)(k-k_0)^{p-j-1} ds  $$

$$|\partial\text{prep}(k,f)/\partial k| \le p|k-k_0|^{p-1}[ 1+ c\sum_{j=1}^p f^{-(j(1-\epsilon)/p)} + c\sum_{j=1}^{p-1}f f^{((1-\epsilon)/p)(-j+1)} $$
$$\le p|k-k_0|^{p-1}[1+ cf^{\epsilon +(1-\epsilon)/p}(1+ O(f^{(1-\epsilon)/p})].$$
It follows that 

$$|f\partial  \log \text{prep}(k,f)/\partial k| \le pf|k-k_0|^{-1}(1+ cf^{\epsilon}).$$
We easily find $|k_2^3 -k_1^3| \le (3f/4)(C + \log1/f) $ so that $k_2 = k_1(1+(k_2^3 - k_1^3)/k_1^3)^{1/3}$.  Thus
$$|k_2-k_1| \le c'(f\log1/f)(k_0/2)^{-2} = Cf\log1/f.$$
Let $G(k,f) = \log m(k,f) - \log \text{prep}(k,f)$.  We have $$k_n^3 - k_{n-1}^3 = (3if/4)(k_{n-1} - k_{n-2})\int_0^1 \frac{\partial G}{\partial k} (k_{n-2} + t(k_{n-1} - k_{n-2},f)dt $$
Thus 
\begin{align} \label{kdiff}
&|k_n^3 - k_{n-1}^3|\le (3f/4)|k_{n-1} - k_{n-2}|\int_0^1 (c + 2p |k_{n-2} + t(k_{n-1} - k_{n-2})|^{-1})dt \nonumber \\
&\le (3f/4)(c|k_{n-1} - k_{n-2}| + 2p\log(1/(1-|k_{n-1} - k_{n-2}|k_{n-2}^{-1})).
\end{align}
Let us assume the Weierstrass preparation theorem holds for $|k-k_0|\le \delta < k_0/2$ and that $m(k,f)$ is analytic (in $k$ and non-zero in this ball for small $f$. Assume $|\lambda_f(j)| \le \delta/3$. 
Let us make the inductive hypotheses that $3k_0/2 \ge |k_l| \ge k_0/2$ and $|k_{l+1} - k_l| \le (C_0f)^{l-1}|k_2-k_1|$ for $1\le l \le n-2$ where we take $C_0 = 2(c + 2p)/k_0^2$.   Then $$k_1 - \sum_{l=1}^{n-2}|k_{l+1} - k_l| \le |k_{n-1}| \le k_1 + \sum_{l=1}^{n-2}|k_{l+1} - k_l|.$$
It follows that $$k_1 - (1-C_0f)^{-1}|k_2-k_1| \le |k_{n-1}| \le k_1 + (1-C_0f)^{-1}|k_2-k_1|$$
and thus for small $f$, $2k_0/3 \ge |k_{n-1}| \ge k_0/2$. Using (\ref{kdiff}), the lower bound on $|k_{n-1}|$, and the induction hypothesis we obtain for $0<f < f_0$ with $f_0$ independent of $n$
$$|k_n - k_{n-1}| \le (f/4|k_{n-1}|^2)(2c+4p)|k_{n-1} - k_{n-2}| \le C_0f|k_{n-1} - k_{n-2}|.$$
The induction is complete.  This estimate shows that for $f$ sufficiently small, $k=k(j) = \lim_{n\to \infty} k_n$ exists and satisfies (\ref{peqn}).

Actually $k_2(j)$ is close enough to the limit to get a good idea of what the string of resonances looks like near $k_0$ (but not too near). Thus

$$k(j) = k_1(j) -i(f/4k_1^2)\Big(\log m(k_1(j),0) - \log \text{prep}(k_1(j),f)\Big) + O((f\log1/f)^2)$$ 
$$k_1(j) = k_0 + \lambda_f(j)= (3\pi f j/2)^{1/3}$$
$$|\lambda_f(j)| \ge 2f^{(1-\epsilon)/p} $$
Since $|k_n - k_0| \ge |\lambda_f| -  \sum_{j=2}^n |k_j-k_{j-1}| \ge  2f^{(1-\epsilon)/p} - (1-C_0f)^{-1}|k_2-k_1| \ge f^{(1-\epsilon)/p}$
for small enough $f$, we have $|\text{prep}(k(j),f)| \ge f^{1-\epsilon}$.  We are ready to state and prove 

\begin{theorem} \label{orderp}
Suppose the reflection coefficient for scattering vanishes at $k_0 > 0$ of order $p$. Then the quantities $k(j)$ given above define a string of resonances of $H_f$.  Suppose $k$ is a resonance of $H_f$.  Then given $\epsilon \in (0,1)$ and $\delta > 0$ small enough with $k_0/2 > \delta > |k-k_0| > f^{(1-\epsilon)/p}$  there exists a positive integer $j$ such that $k=k(j)$.
\end{theorem}

\begin{proof}
Suppose that $k$ is a solution to (\ref{peqn}) satisfying $|k-k_0|< \delta$ and $|k-k_0|> f^{(1-\epsilon)/p}$. Then there is an integer $j$ such that $$k^3 = 3\pi f j/2 -(3if/4)(\log m(k,f) - \log \text{prep}(k,f)).$$ Without loss of generality we take the same branches of the logarithms we took in defining $k(j)$.
We obtain $$k^3 - k(j)^3 = -(3if/4)\Big((\log m(k,f) - \log m(k(j),f) - \log \text{prep}(k,f)) +\log \text{prep}(k(j),f))\Big).$$ At this point we are not able to get a good estimate for the difference of the final two terms involving $\text{prep}$ so we just estimate the difference by the absolute value of the sum.  Thus we obtain 
$$|k^3 -k(j)^3| \le Cf(|k-k(j)| + \log1/f)).$$
Factoring $k^3 - k(j)^3$ and using the fact that both $k$ and $k(j)$ are close to $k_0$ we obtain
$$|k-k(j)|\le C'f(|k-k(j)| + \log 1/f)$$ so that $|k-k(j)| \le C''f\log 1/f$.  Using this rough estimate we can now estimate the difference of the $\log\text{prep}$ terms using $$\log \text{prep}(k,f)) -\log \text{prep}(k(j),f)) = (k-k(j))\int_0^1 \partial  \log \text{prep}(k(j)+ t(k-k(j),f)/\partial kdt.$$ Here we use $|k(j) + t(k-k(j)) - k_0| \ge (1/2)f^{(1-\epsilon)/p}$ so that 

$$|k-k(j)| \le C|k-k(j)|(f + f^{(p -1 +\epsilon)/p})$$ which implies $k=k(j)$.
\end{proof}

\section{Resonances near the line $\arg k = -\pi/3$ } \label{line}

Recalling the definitions of $A_1$ and $A_2$ at the beginning of the previous section,

we know that $k$ is a resonance if there is a non-zero solution $\psi(x)$ to $-\psi'' + (V+f x - k^2)\psi =0$ satisfying
\begin{align}
\label{psiA1}\frac{\psi'(-L)}{\psi(-L)} &= \frac{A_1'(-L)}{A_1(-L)} \\
\label{psiA2}\frac{\psi'(L)}{\psi(L)} &= \frac{A_2'(L)}{A_2(L)}. 
 \end{align}
 
 In the last section where $k$ is close to a point $k_0$ on the positive real axis, $A_1$ and $A_1'$ had single sum asymptotic expansions for $f\downarrow 0$ while  $A_2$ and $A_2'$ have double sum expansions. There we could define $\psi$ to be the solution satisfying \eqref{psiA1} and use \eqref{psiA2} as the equation that determined the resonances. In this section we consider $k$ close to $k_0$ with $\arg k_0 = -\pi/3$. Then the situation is reversed: Now $A_2$ and $A_2'$ have single sum asymptotic expansions while $A_1$ and $A_1'$ have double sum expansions. We will now take $\psi$ to be the solution satisfying \eqref{psiA2} and use \eqref{psiA1} as  our resonance defining equation. Since we need more information than is provided in the expansions in [1], we will use the approximation in Appendix 10 which proves properties of the error term. 

So let us define $\psi(x; k, f, \alpha)$ to be the solution of $-\psi'' + (V+f x - k^2)\psi =0$ satisfying $\psi(L)=1$ and $\psi'(L)=\alpha$. Let $\lambda(k,f) = -i A_2'(L)/A_2(L)$. Then $\psi(x,k,f,i\lambda(k,f))$ is the solution satisfying \eqref{psiA2}. 
\begin{proposition}
Let $k$ be close to $k_0$ where $\arg k_0 =-\pi/3$. Then 
$$
\lambda(k,f) = k + a(f,k)
$$
where $a(f,k)$ is analytic in $k$ near $k_0$, $C^\infty$ in $f$ for $f\ge 0$ and $O(f)$ as $f\downarrow 0$.
\end{proposition}
\begin{proof}
Let $w_2=f^{1/3}(-L-k^2/f)$. When $f\downarrow 0$, $w_2$ avoids a sector about the negative real axis.
We can therefore use the asymptotic formulas 
\begin{align}
\label{AiA1}\Ai(w_2) &= \frac{e^{-\zeta}}{2\sqrt{\pi}w_2^{1/4}}\left(1+a_1(k,f)\right)\\
\label{AiA2}{\Ai}'(w_2) &= -\frac{w_2^{1/4} e^{-\zeta}}{2\sqrt{\pi}}\left(1+a_2(k,f)\right),
\end{align}
where 
\begin{equation}\label{zetadef}
\zeta=(2/3)w_2^{3/2} =-i \frac{2k^3}{3f}\left(1+\frac{fL}{k^2}\right)^{3/2},
\end{equation} 
and by Appendix 9 the error terms $a_1$ and $a_2$ are analytic in $k$ near $k_0$, $C^\infty$ in $f$ and $O(f)$, and

The proposition follows easily from this.
\end{proof}

Now we analyze $A_1'(-L)/A_1(-L)$. We have
$$
A_1(-L) = \Ai(w_1), \quad A_1'(-L) = e^{2\pi i/3} f^{1/3}{\Ai}'(w_1)
$$
where
$$
w_1=e^{2\pi i/3} f^{1/3}(-L-k^2/f).
$$
When $\arg k = -\pi i/3$, then as $f\downarrow 0$, $w_1$ moves to infinity along the negative real axis.  We may use the 
asymptotic formulas for $\Ai$ and $\Ai'$ in Appendix 10:

with $k_1=e^{i\pi/3}k$ and $\eta = k_1^2(1+f L/k^2)$
\begin{align*}
\frac{A_1'(-L)}{A_1(-L)} &= 
\eta^{1/2}
\frac{ e^{2i \eta^{3/2}/(3f)}(-i\tilde a_1) + e^{-2i \eta^{3/2}/(3f)}(-\tilde a_2)}{e^{2i \eta^{3/2}/(3f)}(a_1) + e^{-2i \eta^{3/2}/(3f)}(i a_2)},
\end{align*}
where 
each $a_i(k,f)$ and $\tilde a_i(k,f)$ is smooth in $f$, analytic in $k$  and equal to $1$ when $f=0$.
We know that $k$ near $k_0$ is a resonance exactly when \eqref{psiA1} holds, that is, if the left side of this equation is equal to $\psi'(-L)/\psi(-L)$. 
If this condition holds we can solve the linear fractional transformation for $e^{-4i\eta^3/3f}$, and we find that $k$ is a resonance when
\begin{equation} \label{likenopole}
e^{-4i\eta^3/3f}= iG(\eta)(1+b(\eta,f))
\end{equation}
where the analytic function
$G$ is a multiple of $-\frac{\psi' - ik \psi}{\psi'+ ik\psi}$ 
when $f = 0$ and $b(\eta,f)$ is  smooth in $f$, analytic in $k$   and $O(f)$ as $f\downarrow 0$.
This quantity cannot have a zero at $k=k_0$ since according to (\ref{nopole}), near $k_0$ $e^{-4i\eta^3/3f} = -2k/(F(k) + O(f))$ and $F(k)$ is analytic near $k_0$.

Note that (\ref{likenopole}) is exactly the same equation as \eqref{errornopole} which arose when we found the resonances near the positive real axis in a neighborhood of a point where the reflection coefficient did not vanish. 
As before it has a solution given by the fixed point of a contraction.
Thus setting $\log iG(\eta) = - l(\eta) + i \theta(\eta)$ we have
\begin{theorem}
Suppose $k_0 = e^{-i\pi/3}\eta_0$ with $\eta_0 > 0$  a point where $G$ does not have a pole.  Define $\eta_0(j) = (3\pi f j)^{1/3}$. Then the resonances near $k_0$ are given by $k(j) = e^{-i\pi/3}\eta(j)$ with
$$\eta(j) =\eta_0(j) - 4^{-1}\Big(\theta(\eta_0(j))/\eta_0(j)^2 + il(\eta_0(j))/\eta_0(j)^2\Big)f + O(f^2).$$ 
\end{theorem}

The situation where $\frac{\psi' - ik \psi}{\psi'+ ik\psi}$ has a pole at $k_0$ can be treated in essentially the same way as was done in Theorem \ref{orderp}.

\section{Higher dimensions} \label{ndimensions}

We generalize equation (\ref{K_c}).  Writing $H_{0,f} = p_1^2 + fx_1 + p_{\perp}^2$, we write $ L^2(\mathbb {R}^d)$ as a direct integral of functions of the perpendicular momentum with values in $L^2(\mathbb{R})$.
We assume the potential $V$ is factorized as $V = V_1V_2$ where the $V_j$  are real, bounded, measurable with compact support in $B_R = \{x\in \mathbb {R^d} :|x| <R\}$.  Thus the weighted resolvent of $H_{0,f}$ has an integral kernel 

$$K_{0,f}(x_1,y_1; p_\perp,z)  = f^{-1/3}V_1(x_1,\cdot) \Big(\int \Ai(f^{1/3}(x_1-t))(t-(z-p_\perp^2)/f)^{-1}\Ai(f^{1/3}(y_1-t))dt\Big) V_2(y_1,\cdot).$$
We now analytically continue from the upper to the lower half plane and find the kernel of $(V_1(H_{0,f} - z)^{-1}V_2)_c$

$$ K_{0,f,c} = \tilde Q + K_{0,f}$$
where $\tilde Q$ is the integral kernel in the variables $(x_1,y_1)$ of an operator $Q$: 

\begin{equation} \label{Qtilde}
\tilde Q(x_1,y_1)  = 2\pi i f^{-1/3}V_1(x_1,\cdot) \Ai(f^{1/3}(x_1 -(z-p_\perp^2)/f)\Ai(f^{1/3}(y_1 -(z-p_\perp^2)/f)V(y_1,\cdot).
\end{equation}
As $f \downarrow 0$, 

$$\Ai(f^{1/3}(x_1 -(z-p_\perp^2)/f) = \frac{f^{1/6}}{2\sqrt{\pi \eta}}e^{2i\eta^3/3f} e^{-i\eta x_1 -i\pi/4}(1+O(f))$$
where $\eta = \sqrt{z- p_\perp^2}$.  Here $\text{Im}z < 0$.  In the following we think of using the Fourier transform to diagonalize $p_\perp$, so $p_\perp$ becomes $\xi_\perp$.

First assume that $\arg z \in [-\pi, -2\pi/3)$.  If $\arg (z-\xi_\perp^2) = - \pi + \phi$, with $\phi \in [0, \pi/3 -\delta]$,  $\text{Re}\ i \eta^3= -|\eta|^3 \cos (3\phi/2) = - |\eta|^3 \sin(3\delta/2)$.  Taking into account the compact support of $V_j$ which we assume to be in the ball $B_R = \{x: |x| < R\}$ and assuming $\text{Im} z < - \delta <0 $ we have for $|x_1| < R$

$$|\Ai(f^{1/3}(x_1 -(z-\xi_\perp^2)/f)| \le Cf^{1/6}|\eta|^{-1}e^{-\frac{2 |\eta|^3\sin(3\delta/2)}{3f} + |\eta|R}$$
$$\le c_1 f^{1/6}e^{-c_2/f}$$ for some positive constants $c_j$.
We thus have
\begin{theorem}
Given $\delta > 0$, if $f$ is small enough there are no resonances of $H_f$ in the region $\{z: \arg z \in [-\pi + \delta, - 2\pi/3 - \delta]  ,\text {Im} z \le - \delta < 0\}$.  
\end{theorem}

\begin{proof}

Resonances are points $z$ so that $\ker (I + Q(I - V_1(H_f-z)^{-1}V_2))$ is not $\{0\}$.  From the above estimates

$$|(\psi, Q\phi)| = | \int (V_1\psi(x_1), F(x_1,y_1)V_2\phi(y_1))dx_1dy_1|$$
$$\le  \int ||V_1\psi(x_1)|| c_1^2 e^{-2c_2/f}||V_2\phi(y_1)||dx_1dy_1$$
where $F(x_1,y_1)$ is $\tilde Q $ with the $V_j$'s removed and we have used
$$||F(x_1,y_1)|| \le c_1^2 e^{-2c_2/f}$$ for $|x_1|, |y_1| < R$.
Thus by the Schwarz inequality
$$|(\psi, Q\phi)| \le R ||V_1\psi|| ||V_2\phi|| c_1^2 e^{-2c_2/f} $$
which implies  $(I + Q(I - V_1(H_f-z)^{-1}V_2)$ is invertible for small $f$.
\end{proof}

Now consider the region $\arg z \in (-2\pi/3 + \delta, -\delta)$ for small $\delta >0$.  This is a region where progress should be reasonably simple compared to the remaining regions along the positive real axis and the line $\arg z = - 2\pi/3$. Unfortunately we have nothing to report about the existence of resonances in this region. But we give some information about the operator which needs to be examined to make further progress.  

We modify $\tilde Q$ slightly to define an operator $M$:

$$M\psi(x_1,\cdot) = 2\pi i f^{-1/3}1_{[-R,R]}(x_1) \int \Ai(f^{1/3}(x_1 + (p_\perp^2 - z)/f))\Ai(f^{1/3}(y_1 + (p_\perp^2 - z)/f))1_{[-R,R]}(y_1)\psi(y_1,\cdot)dy_1.$$
$M$ has the virtue that it is a multiple of a projection:

\begin{align*}
& M^2 \psi(x_1,\xi_\perp) = \\
&(2\pi if^{-1/3})^21_{[-R,R]}(x_1) \Ai(f^{1/3}(x_1 + (\xi_\perp^2 - z)/f))\int \Ai(f^{1/3}(y_1 + (\xi_\perp^2 - z)/f))\\
&1_{[-R,R]}(y_1)\Ai(f^{1/3}(y_1 + (\xi_\perp^2 - z)/f))\int \Ai(f^{1/3}(w_1 + (\xi_\perp^2 - z)/f))\\
&1_{[-R,R]}(y_1)\psi(w_1,\xi_\perp)dw_1 dy_1\\
&= F(\xi_\perp^2) M\psi(x_1,\xi_\perp) 
\end{align*}
where with  $\eta = \sqrt {z- \xi_\perp^2}$,
$$F(\xi_\perp^2) = 2\pi i f^{-1/3}\int_{-R}^R \Ai(f^{1/3}(x_1 + (\xi_\perp^2 - z)/f))^2 dx_1 = \frac{e^{4i\eta^3/3f}}{2\eta}(\sin(2\eta R)/\eta)(1+ O(f)).$$
Thus $M = F(p_\perp^2) P_f$ where $P_f^2 = P_f$ and $P_f$ commutes with $F(p_\perp^2)$.  Even though we have not indicated it, of course $F(p_\perp^2)$ also depends on $f$ and $z$.  

We note that $Q = V_1MV_2$ so using the fact that $\sigma(AB)\setminus \{0\} =\sigma(BA)\setminus \{0\}$ we have 
\begin{proposition}
The resonances of $H_f$ in the lower half plane are the points $z$ such that the operator
$$F(p_\perp^2) P_f (V-V(H_f -z)^{-1}V)P_f$$
has eigenvalue $-1$.
\end{proposition}

We give some of the asymptotics of $P_f$ and $F(\xi_\perp^2)$:

\begin{lemma}
For $\text{Im} z < 0$,  $||P_f - P_0|| = O(f)$ where
$P_0$ has kernel as an operator on $L^2$ functions of the perpendicular momentum and the $x_1$ variable
$$1_{[-R,R]}(x_1)e^{-i\eta x_1}(\sin{2\eta R}/2\eta )^{-1}e^{-i\eta y_1}1_{[-R,R]}(y_1).$$
Here $\eta = \sqrt{z -\xi_\perp^2}$ is always in the 4th quandrant.  

The function $F(\xi_\perp^2)$ satisfies the following estimates when $\text{Re}z + |\text{Im}z|/\sqrt3 >0$ and $\text{Im}z<0$.
Set $\alpha^2 = \text{Re}z + |\text{Im}z|/\sqrt 3$.  Then 

$$|F| \ge c_1e^{(2/\sqrt 3)^{1/2}|\text{Im}z|^{1/2}(\alpha^2 - \xi_\perp^2)/f}$$

if $\xi_\perp^2 < \alpha^2$.

$$|F| \le c_2e^{-(2/\sqrt 3)^{1/2}|\text{Im}z|^{1/2}(\xi_\perp^2- \alpha^2)/f}$$

if $\beta > \xi_\perp^2 > \alpha^2$.

$$|F| \le c_3e^{-(1/\sqrt 3)^{1/2}|\text{Im}z|^{1/2}(\xi_\perp^2- \alpha^2)/f}$$

if $\xi_\perp^2 > \beta$.
\end{lemma}
\flushleft
Note that $\arg (z -\xi_\perp^2) = -2\pi/3$ exactly when $\alpha^2 = 0$.  When $z-\xi_\perp^2$ crosses the line with argument $-2\pi/3$ we cross from the region where $F$ blows up as $f \downarrow 0$ to the region where it decays exponentially.

\section
{Appendix - The Airy function $\Ai(z)$ in the sector $\arg z \in (-2\pi/3 + \delta, -\delta)$}\label{AppendixAiry}

We believe the results in this appendix are known but could not find a suitable reference.

\begin{theorem} With $\eta=\sqrt w$, $\arg \eta \in (-\pi/2,\pi/6)$, $\sqrt {\eta}e^{2\eta^3/3}\Ai(w)$ extends to an analytic function, $J$, of $\zeta^{-1}, \zeta = 2\eta^3/3$ with $\arg \zeta^{-1} \in (-\pi/2, 3\pi/2)$. The function $J$ and all its derivatives extend continuously to the origin in this sector. 
\end{theorem}  

\begin{proof}
We start with the integral representation $\Ai(w) = \lim_{R \to \infty} (2\pi )^{-1}\int_{-R}^R e^{i(sw+s^3/3)}ds$ where $w\ge 0$ and deform the contour to $s = i\sqrt w + u$ with $u \in \mathbb{R}$ to obtain

$$\Ai(w) = \frac{e^{-2\eta^3/3}}{2\pi} \int_{-\infty}^{\infty} e^{-\eta u^2}e^{iu^3/3}du $$ with $\eta = \sqrt w$.  We can analytically continue to $\arg \eta \in (-\pi/2, \pi/2)$.   Thus using polar coordinates with $\zeta ^{-1} = 3w^{-3/2}/2 = r e^{i\theta}$, we have $$2\pi \sqrt \eta e^{2\eta^3/3}\Ai(w) = e^{-i\theta/6}\int_{-\infty}^{\infty} \exp(-e^{-i\theta/3}t^2)g(t^6r)dt$$ where $g(\psi) = \cos(\sqrt {2\psi/27})$ is a $C^{\infty}$ function of its argument.  With an integration by parts it is easy to verify the polar form of the Cauchy-Riemann equations ($(\partial/\partial r +i r^{-1} \partial/\partial \theta)J = 0$).  We take $\theta  = \arg \zeta^{-1} \in (-\pi/2, 3\pi/2)$.  Because $J$ satisfies the Cauchy-Riemann equations $J^{(n)}$ is given by $e^{-in\theta}\partial^n J/\partial r^n$ which is easily shown to have a limit as $r\to 0$ which is independent of $\theta$.
\end{proof}
Thus in particular $J$ has an asymptotic expansion 
$$J(\zeta^{-1}) \sim \sum_{n=0}^\infty c_n \zeta^{-n}$$ which can be differentiated term by term.

\section{Appendix - The Airy function $\Ai(z)$ in the sector $\arg z \in [-\pi, -\pi + \delta]$}\label{AppendixAiry2}

We follow \cite{HR} but the latter reference does not go far enough for our purposes.  We write the Airy function $\Ai(w)$ for $w$ real as  $$\Ai(w) = (2\pi)^{-1}\lim_{R \to \infty}\int_{-R}^R e^{-i(t^3/3 +wt)}dt.$$ We set $w = f^{1/3}(L-z/f) = -(z/f^{2/3})(1-fL/z)$ where we first keep $z$ real and positive.  We set $\eta = z(1-fL/z) = z- fL$. With a change of variable we have 
\begin{equation} \label{sintegral}
\Ai(w) = (2\pi f^{1/3})^{-1}\int_{\mathcal{C}} e^{(-i/f)(s^3/3 - \eta s)}ds.
\end{equation}

Here we have used the fact the $\eta$ is real to move the contour into the lower half plane so that $\mathcal{C}$ is just $s(t) = t - i\alpha$ with $t\in \mathbb{R}$ and $\alpha > 0$.  We then see that the Airy function is analytic in $\eta$ for $\eta$ in a neighborhood of a real point. We distort the contour further.  We use steepest descents near the critical points of the exponential where $(s^3/3 - \eta s)' = s^2 - \eta = 0$, namely the points  $s = \pm \sqrt{\eta}$.   Thus near $+\sqrt{\eta}$  we write $s = \zeta + \sqrt{\eta}$.  Note that at this critical point we have $s^3/3 - \eta s = -2\eta^{3/2}/3$ and 
$-i(s^3/3 - \eta s + 2\eta^{3/2}/3) = -i(\zeta^3/3 +\sqrt{\eta}\zeta^2)$. The steepest descents curve near $\zeta = 0$ will come from setting $\text{Re}(\zeta^3/3 + \sqrt{\eta}\zeta^2) = 0$.  %Before distorting the contour further note that if we differentiate the quantity $e^{-i({s^3/3 - \eta s + 2\eta^{3/2}/3)/f$ with respect to $\eta$, we bring down a factor of $i( s- \sqrt {\eta})/f$ which equals $i\zeta/f$ on the part of the contour near $\sqrt{\eta}$.  
We will use part of the following contour for $s$ near $\sqrt{\eta}$:  We solve the equation $\text{Re}(\zeta^3/3 + \sqrt{\eta}\zeta^2) = 0$ and find with $\zeta = x + iy$, 

$$y = \frac{-x(\gamma + x/3)}{\nu + \sqrt{\nu^2 +(\gamma + x)(\gamma + x/3)}}$$

We have $$-i(\sqrt{\eta} \zeta^2 + \zeta^3/3) = -b^2x^2 + \sum_{n = 3}^\infty B_n x^n $$ where $b^2 = 2\gamma^{-2} |\eta|(\sqrt{|\eta|} -\nu)$.  The series converges for $|x| < \gamma$ (see \cite{HR}).  We have the following estimates in the indicated regions:
\begin{align*}
-i(\zeta^3/3 + \sqrt{\eta}\zeta^2)&\le -(\gamma/3)^2x^2/\sqrt{\nu^2 + (\gamma +x)(\gamma +x/3)}; 0 \ge x \ge - \gamma, \nu \ge 0 \\
&\le -(\gamma^2/18|\eta|)x^2; 0 \le x \le 4\gamma, \nu \ge 0 \\
&\le -\gamma x^2 ; 0 \le x \le 10\gamma, 0 \ge \nu \ge -\gamma/10 \\ 
&\le -\gamma x^2;  - \gamma \le x \le 0,  0 \ge \nu.
\end{align*}

For $s$ near $-\sqrt{\eta}$ we write $s=\zeta - \sqrt{\eta}$ which gives $-i(s^3/3 -\eta s - 2\eta^{3/2}/3) = -i(\zeta^3/3 - \sqrt{\eta}\zeta^2)$.  Setting $\text{Re}(\zeta^3/3 - \sqrt{\eta}\zeta^2) = 0$ and $\zeta = x+ iy$ we find for 

$$y = x(\nu + \sqrt{\nu^2 + (\gamma-x)(\gamma-x/3)})(\gamma-x)^{-1}.$$

We find 
$$-i(\zeta^3/3 - \sqrt{\eta}\zeta^2) = -\tilde{b}^2 x^2 + \sum_{n=3}^{\infty}\tilde{B}_nx^n$$
where $\tilde{b}^2 = 2\gamma^{-2} |\eta|(|\eta|^{1/2} + \nu)$.  The series converges for $|x| < \gamma$.

We have the following estimates in the indicated regions:
\begin{align*}
-i(\zeta^3/3 - \sqrt{\eta}\zeta^2) &\le -\gamma x^2 \sqrt{\frac{\gamma - x/3}{\gamma -x}}; 0 \le x < \gamma, \nu \ge 0 \\
&\le -\gamma x^2 (y/x) \le - \gamma x^2 \frac{\gamma - x/3}{2\sqrt{\nu^2 + (\gamma-x)(\gamma-x/3)}}; 0 \le x < \gamma, \nu < 0 \\
&\le -\gamma x^2/5; -10\gamma \le x \le 0, 0 \le \nu < \gamma \\
&\le -(2\gamma/5)x^2; - 8\gamma \le x \le 0, - \gamma/3 \le \nu \le 0.
\end{align*}

We write $$\Ai(w) = (2\pi f^{1/3})^{-1}(\int_{\mathcal{C_-}} e^{(-i/f)(s^3/3 - \eta s)}ds  + \int_{\mathcal{C_+}} e^{(-i/f)(s^3/3 - \eta s)}ds).$$

We first consider the integral over the part of the contour $\mathcal{C_+}$ near $\sqrt{\eta}$:  

$$\int_{C_+'} e^{(-i/f)(s^3/3 - \eta s)}ds = (2\pi f^{1/3})^{-1}e^{(2i\eta^{3/2}/3f)}\int_{\mathcal{B}_+} e^{(-i/f)(\zeta^3/3 + \sqrt{\eta} \zeta^2) }d\zeta.$$
To obtain the part of the contour in the $s$ variable near $\sqrt{\eta}$ we take $s = \zeta + \sqrt{\eta}$.  This gives us the contour  $\mathcal{B}_+$ with $\zeta = x + iy $ where $x \in (\-\beta, \infty) , y = - \alpha$ , where as mentioned above $\alpha > \nu$.  We thus see that this integral is analytic in $\sqrt{\eta}$ as long as $f>0$.  We will see that after multiplying by $f^{-1/2}$ we can take the limit as $f\to 0$ uniformly for $\eta$ near a point $k_0 >0$. We deform the contour to obtain a new contour $\Gamma_+$ where $$\zeta = x+iy, x \in (-\gamma/2, \gamma/2), y =  \frac{-x(\gamma + x/3)}{\nu + \sqrt{\nu^2 +(\gamma + x)(\gamma + x/3)}}.$$ 

Thus we have 
$$I_+ = (2\pi f^{1/3})^{-1}\int_{\Gamma_+}e^{-i(\zeta^3/3 +\sqrt{\eta}\zeta^2)/f} d\zeta = \frac {f^{1/6}}{2\pi} J_+ $$
$$J_+ = f^{-1/2}\int_{-\gamma/2}^{\gamma/2} e^{f^{-1}(-b^2x^2 + \sum_{n=3}^{\infty}B_nx^n)}(1+ idy/dx)dx$$ 
We expand the integrand in a convergent power series keeping the $e^{-b^2x^2}$ to obtain

$$J_+ = \int_{-\gamma f^{-1/2}}^{\gamma f^{-1/2}} e^{-b^2x^2}(1 + \sum_{k=1}^\infty(\sum_{n=3}^\infty B_nf^{-1 + n/2}x^n )^k/k!)(1+i\sum_{m=0}^\infty a_m f^{m/2} x^m)dx$$
The odd powers of $x$ do not contribute and thus we have 
$$J_+ = 2(1+ia_0)\int_0^{\gamma f^{-1/2}}e^{-b^2x^2}dx + \sum_{m + |n|/3 \ge 1, l \ge 1 }  c_{n,m}f^{(|n| + m -2l)/2}  \int_0^{\gamma f^{-1/2}}e^{-b^2x^2} x^{|n|+ m}dx$$
Here $n$ is a multi-index $n = (n_1,n_2, \cdots, n_l)$ and $|n| = n_1 + \cdots + n_l$.  Each $n_j \ge 3$ and thus the power of $f$, $(|n| + m)/2 - l $ is positive and since $|n| + m$ is even this power is a positive integer $\ge (l+m)/2$ . Note that $\int_{0}^ {\gamma f^{-1/2}}e^{-b^2x^2} x^{|n| + m}dx$ is $C^{\infty}$ in $f$ near $0$ as long as we define $f^{-n} e^{-b^2\gamma^2/f}$ to be $0$ at $f=0$.  Thus $J_+$ is $C^\infty$ for $f\ge 0$ in the sense that the derivatives have limits as $f \downarrow 0$.  In the limit $f\downarrow 0$ the first term can be calculated (with some effort - see also \cite{HR}) to be $e^{-i\pi/4}\sqrt{\pi/\eta^{1/2}}$.

\vspace{.5cm}

We now connect up this curve to infinity as follows:  We take $x+iy = \gamma/2 + t -i y(\gamma/2)$ with $t \in [0, \infty)$. For simplicity take $|\nu| < \gamma/10 $.  It is then easy to see that the real part of $-i(\zeta^3/3 +\sqrt{\eta}\zeta^2)\le - C < 0$ independent of $t\ge 0 $ and independent of $\eta$ for $\eta$ near a real point $k_0 > 0$. If the extended curve is called $\tilde{\Gamma}_+$ This is easily seen to imply that $$\tilde{I}_+ = (2\pi f^{1/3})^{-1}\int_{\tilde{\Gamma}_+}e^{-i(\zeta^3/3 +\sqrt{\eta}\zeta^2)/f} d\zeta = \frac {f^{1/6}}{2\pi} \tilde{J}_+ $$ with $\tilde{J}_+$ a $C^{\infty}$ function in $f$ for $f_0 >f\ge 0$ and analytic in $\sqrt{\eta}$ for $\sqrt{\eta}$ near a point $k_0 > 0$.

\vspace{.5cm}

The curve near the critical point $-\sqrt{\eta}$ can be handled similarly.  We set $s= \zeta - \sqrt{\eta}$ giving $-i(s^3/3 - \eta s) = -i(\zeta^3/3 - \sqrt{\eta}\zeta^2) -2i\eta^{3/2}/3$ = and $\zeta = x+ iy$.  We demand that for $x\in (-\gamma/2, \gamma/2)$ we have  $\text{Re}[-i(\zeta^3/3 - \sqrt{\eta}\zeta^2)] = 0$.  This results in 
$$y = \frac{x(\gamma - x/3)}{-\nu + \sqrt{\nu^2 + (\gamma-x)(\gamma-x/3)}}$$.  

We connect this curve to infinity and to the imaginary axis just as in the case where we extended the curve $\Gamma_+$.  The result is the same.  It remains to connect the two curves on the imaginary axis.  The curve on the left ends on the imaginary axis at $s = i(2\nu + \sqrt{\nu^2 + 5\gamma^2/12})$ or $\zeta = \gamma + i(\nu + \sqrt{\nu^2 + 5\gamma^2/12}) $.  The curve on the right ends on the imaginary axis at $s = i(-2\nu + \sqrt{\nu^2 + 5\gamma^2/12})$.  Thus we extend the curve on the left as $\zeta = \gamma + i(u+ \nu + \sqrt{\nu^2 + 5\gamma^2/12})$ or $s = \gamma + i( u + 2\nu + \sqrt{\nu^2 + 5\gamma^2/12})$ with $u$ going from $0$ to -$4\nu$.  It is not hard to see that this piece is $C^\infty$ in $f$ for $f\ge 0$ with the function and all derivatives equal to zero at $f=0$.  (Of course  the derivative at zero is the right hand derivative.)  It is convenient in the estimates to restrict $|\nu| < \gamma/10$. Even though the contours involve $\sqrt{\eta}$ in a non-analytic way, it is not hard to see that they can be distorted in such a way that given $\sqrt{\eta_0}$ the contour can be chosen to depend on this quantity but then $\sqrt{\eta}$ can be varied in a small disk around this point and the result is analytic in the disk with estimates similar to what we have derived.

Thus what we have shown is that 
$$\Ai(w) = f^{1/6} e^{-i\pi/4}\frac{1}{2 \sqrt{\pi\eta^{1/2}}}(e^{2i\eta^{3/2}/3f} a_1(k,f) + ie^{-2i\eta^{3/2}/3f} a_2(k,f))$$

where $a_j(k,f)$ is $C^{\infty}$ in $f$ for $f_0 > f \ge 0$ and analytic in $k$ for $k$ near a point on the positive real axis. Since the limits involved in taking derivatives in $f$ converge uniformly in $k$, these derivatives are also analytic in $k$ in a neighborhood of a positive real point.  We have $a_j(k,0) = 1$.

\vspace{.2cm}

We also need the derivative of $\Ai(w)$ with respect to the spatial variable which we have called $L$ in this appendix. Note that in (\ref{sintegral}) $L$ occurs only in $\eta = z - fL$ so that differentiation with respect to $L$ brings down $-is$ in this integral. In the integral near $+\sqrt{\eta}$ the change of variable is $s=\zeta + \sqrt{\eta}$.  $-i\zeta$ contributes something of at most order $f$ while the term $-i\sqrt{\eta}$ gives the main contribution.  A similar analysis near the critical point $-\sqrt{\eta}$ shows that the main contribution is a factor of $+ i \sqrt{\eta}$.  Thus we obtain

$$(d/dL)\Ai(w) = f^{1/6} e^{-i\pi/4}{\sqrt{\frac{\eta^{1/2}}{4\pi}}}(-ie^{2i\eta^{3/2}/3f} \tilde{a}_1(k,f)  -e^{-2i\eta^{3/2}/3f} \tilde{a}_2(k,f))$$
where $\tilde{a}_j(k,f)$ is $C^{\infty}$ in $f$ for $f_0 > f \ge 0$ and analytic in $k$ for $k$ near a point on the positive real axis. Since the limits involved in taking derivatives in $f$ converge uniformly in $k$, these derivatives are also analytic in $k$ in a neighborhood of a positive real point.  We have $\tilde{a}_j(k,0) = 1$.

We now must do similar estimates with $A_1(x)$ for $z$ near the line $\arg  = -2\pi/3$.  Here the argument of the Airy function $\Ai$ is $w = f^{1/3}(x-z/f)$.  We set $z_1 = e^{2\pi/3} z$ which is near a positive real point and $\eta = z_1(1-fx/z)$ which is also near a point on the positive real axis when $f$ is small. Thus after a change of variable we can write 

\begin{equation} \label{s'integral}
\Ai(w) = (2\pi f^{1/3})^{-1}\int_{\mathcal{C}} e^{(-i/f)(s^3/3 - \eta s)}ds.
\end{equation}

where the contour $\mathcal{C}$ is the same as above.  Thus the same analysis as above works in this situation.  Note that when we differentiate with respect to the spatial variable the $\eta$ dependence on $x$ gives a contribution to the exponential from $-i\eta s/f$ of $ixe^{2\pi i /3}$. Thus differentiation with respect to $x$ brings down a factor of $ie^{2\pi i/3}s$, and thus a main contribution of $\sqrt{\eta}ie^{2\pi i/3}$ for the integral near the critical point $\sqrt{\eta}$.  Similarly there is a main contribution of  $-\sqrt{\eta}ie^{2\pi i/3}$ for the integral near the critical point $-\sqrt{\eta}$.

\end{document}